\documentclass[12pt]{iopart}
\usepackage{graphicx}
\usepackage{subcaption}
\usepackage{xcolor,cite}
\usepackage{amsthm}
\usepackage{amssymb}
\usepackage[utf8]{inputenc}
\newtheorem{theorem}{Theorem}[section]
\newtheorem{definition}{Definition}[section]

\newtheorem*{remark}{Remark}

\begin{document}

\title[Coupled H\'{e}non Map, Part II]{Coupled H\'{e}non Map, Part II: Doubly and Singly Folded Horseshoes in Four Dimensions}

\author{Jizhou Li$^1$, Keisuke Fujioka$^2$, and Akira Shudo$^2$}

\address{$^1$RIKEN iTHEMS, Wako, Saitama 351-0198, Japan}
\address{$^2$Department of Physics, Tokyo Metropolitan University, Tokyo 192-0397, Japan}

\ead{jizhou.li@riken.jp \quad shudo@tmu.ac.jp}
\vspace{10pt}
\begin{indented}
\item[]March 2023
\end{indented}

\begin{abstract}
As a continuation of a previous paper (\textbf{arXiv:2303.05769 [nlin.CD]}), we introduce examples of H\'{e}non-type mappings that exhibit new horseshoe topologies in three and four dimensional spaces that are otherwise impossible in two dimensions.  
\end{abstract}

%
%
%
%
%

\section{Introduction}
In a previous paper \cite{Fujioka23} we have derived a sufficient condition for topological horseshoe and uniform hyperbolicity of a 4-dimensional symplectic map, which is introduced by coupling the two 2-dimensional H\'{e}non maps via linear terms. The coupled H\'{e}non map thus constructed can be viewed as a simple map modeling the horseshoe in higher dimensions. In this paper we further explore new possibilities of horseshoe topologies in three and four dimensions by investigating several new types of coupled H\'{e}non maps. The horseshoes introduced here can only exist beyond two or three dimensions and may serve as templates for the future investigation of horseshoes in multidimensional spaces. 

Two well-known ingredients that give rise to the classical phenomenon of chaos are ``stretching" and ``folding" of phase-space volumes. The stretching creates divergence between nearby initial conditions, and the folding leads to mixing and thus ergodicity of the system. A generic prototype incorporating both ingredients is the Smale horseshoe \cite{Smale65} which acts on a finite region in phase space, expands it along the unstable direction, contracts it along the stable direction, then folds and re-injects it into the original region.  In the due process, mixing is created by folding and re-injection, which leads to nonlinear dynamics that displays typical phenomena of chaos. 

The most well-known example that realizes the Smale horseshoe would be the H\'{e}non map, which is a 2D quadratic map with parameters \cite{Henon76}. There is an elementary proof showing that the horseshoe is realized in a certain parameter space \cite{Devaney79}, and a necessary and sufficient condition has been pursued using sophisticated techniques in the theory of complex dynamical systems \cite{bedford2004real}. The studies of the higher dimensional extension, a class of the coupled H\'{e}non maps, have been done in various directions 
(see for example \cite{mao1988standard,baier1990maximum,ding1990algebraic,todesco1994analysis,astakhov2001multistability,anastassiou2018recent,backer2020elliptic}), 
not only to provide examples of hyperchaos \cite{rossler1979equation} but also to seek the nature of dynamics absent in 2D. 

In physics and chemistry, there is a growing interest in higher-dimensional 
Hamiltonian dynamics. In particular,  the situation in which regular and chaotic orbits coexist 
in a single phase space is often realized not only in astrophysical or chemical reaction dynamics \cite{contopoulos2002order,laskar1994large,wiggins1994normally}
but also in Lagrangian descriptions of fluid dynamics \cite{wiggins2005dynamical}. 
These studies have paid special attention to invariant structures associated with 
regular motions, and related bifurcations in higher-dimensional spaces as well, 
but we may expect a variety of horseshoes with different topologies in chaotic domains.

To the authors' knowledge, 
although the original Smale horseshoe was proposed in the context of arbitrary dimensions, 
it only considered either singly-folded horseshoes or multiply-folded horseshoes with creases 
in the same direction \cite{Smale67}. Therefore, the original horseshoes, even in multidimensional settings, can always be visualized by collapsing the unstable and stable subspaces into one-dimensional unstable and stable directions, respectively, which leads us back to the picture of 2D horseshoes. In this sense, they did not make use of all possible choices of the directions of creases in multidimensional spaces. This leaves one to wonder what would happen if a generalization of the horseshoe folds a multidimensional hyper-cube twice, with creases in mutually independent directions. Although horseshoes in higher dimensions with creases in different directions 
have been illustrated qualitatively \cite{Wiggins88} or modeled by the composition of piecewise linear maps \cite{zhang2019chaotic}, an explicit and minimal form of such maps has not been provided.  

In this article, we propose several mappings which give rise to doubly-folded horseshoes that can only exist beyond two or three dimensions. More specifically, we will introduce a class of H\'{e}non maps that display five different kinds of horseshoe topology, labeled by \textit{Topology I, II, III-A, III-B} and \textit{IV}, respectively, which are:
\begin{enumerate}
\item[a)] Topology I: Singly-folded horseshoe in three dimensions (Section.~\ref{Background}).
\item[b)] Topology II: Doubly-folded horseshoe with independent creases in three dimensions (Section.~\ref{Type II}).
\item[c)] Topology III-A: Doubly-folded horseshoe with independent creases and independent stacking directions in four dimensions (Section.~\ref{Type III-A}).
\item[d)] Topology III-B: Singly-folded horseshoe in four dimensions (Section.~\ref{Type III-B}).
\item[e)] Topology IV: Doubly-folded horseshoe with independent creases and common stacking direction in four dimensions (Section.~\ref{Type IV}).
\end{enumerate}
We emphasize that these five types of topology are by no means comprehensive. Our main purpose here is to list up examples in a heuristic way. These examples would prepare us for the future construction of the entire library of unexplored horseshoe topologies that may give rise to new dynamical phenomena. 

This article is structured as follows. Section.~\ref{Background} reviews the background theory of Smale horseshoe in its original context. Section.~\ref{Type II} introduces a three-dimensional H\'{e}non-type mapping that possesses a doubly-folded horseshoe with creases in mutually independent directions. Section.~\ref{Type III} proposes a coupled H\'{e}non map in four dimensions which, depending on the ranges of parameters, displays either a doubly-folded horseshoe with creases in mutually independent directions and independent stacking directions, or a singly-folded horseshoe. Section.~\ref{Type IV} demonstrates another type of doubly-folded horseshoe in four dimensions which has the same stacking directions after the two foldings. Section.~\ref{Conclusion} concludes the paper and proposes new directions for future study.

\section{Background: Smale horseshoes in two and three dimensions}\label{Background}
In this section we briefly outline the background theory of Smale horseshoe with examples from the H\'{e}non map. The map gives rise to a singly-folded horseshoe whose cross-sections reduce to the well-studied two-dimensional horseshoe, therefore topologically conjugate to a full shift on two symbols. This topology, denoted by \textit{Topology I} hereafter, is the most well-known Smale horseshoe studied in many articles and textbooks. An illustration is given by Fig.~\ref{fig:3D_singly_folded_horseshoe}.  
\begin{figure}
        \centering
        \includegraphics[width=0.6\linewidth]{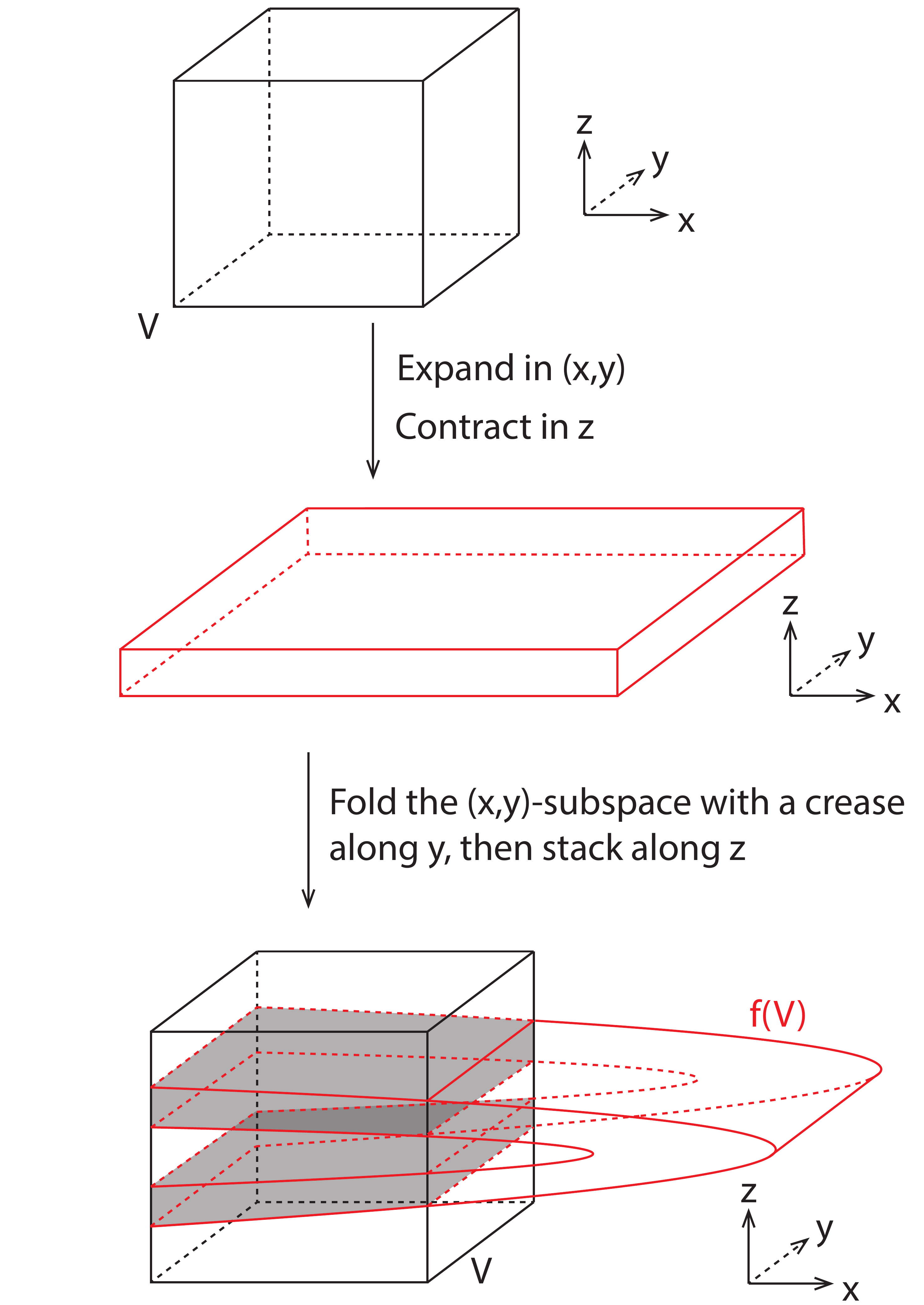} 
        \caption{(Schematic, color online) Topology I: singly-folded horseshoe in three dimensions. Given a cube $V$, the map expands $V$ in the unstable directions ($(x,y)$-plane), contracts it in the stable direction ($z$-axis), then folds it with a crease in the $y$-direction and stacks it along the $z$-direction. The intersection $V\cap f(V)$ is composed of a pair of three-dimensional ``horizontal" slabs, which, when viewed in every $(x,z)$-slice, reduces to a pair of two-dimensional horizontal strips.} \label{fig:3D_singly_folded_horseshoe}
\end{figure}
A concrete realization of Type I is the H\'{e}non map of the form 
\begin{equation}\label{eq:3D singly folded Henon}
\left( \begin{array}{ccc}
x_{n+1}\\
y_{n+1} \\
z_{n+1} \end{array} \right) = 
f_{\rm I} \left( \begin{array}{ccc}
x_{n}\\
y_{n} \\
z_{n} \end{array} \right) =
\left( \begin{array}{ccc}
a_0 - x^2_n - z_n\\
b y_n \\
x_n \end{array} \right)
\end{equation}
where $(x_n,y_n,z_n)^\mathsf{T}$ denotes the position of the $n$th iteration, and the parameters $a_0$, $b$ satisfy  
\begin{eqnarray}
& a_0  > 5+2\sqrt{5} \nonumber \\
& b > 1   \nonumber
\end{eqnarray}
where the bound on $a_0$ is obtained by Devaney and Nitecki in \cite{Devaney79} to realize the horseshoe in the two-dimensional H\'{e}non map, and the bound on $b$ guarantees uniform expansion in $y$. It is trivial to see that the dynamics in $y$ is a constant uniform expansion uncoupled from $(x,z)$. Therefore, on every $(x,z)$-slice we have the two-dimensional H\'{e}non map
\begin{equation}\label{eq:2D Henon}
\left( \begin{array}{ccc}
x_{n+1}\\
z_{n+1} \end{array} \right) = 
\left( \begin{array}{ccc}
a_0 - x^2_n - z_n\\
x_n
 \end{array} \right)\ ,
\end{equation}
which is the one originally proposed in \cite{Henon76} and studied in detail in \cite{Devaney79}. Generalizing the results established by \cite{Devaney79} to the three-dimensional map $f_{\rm I}$, it is straightforward to see that given $a_0>5+2\sqrt{5}$, a cube $V$ can be identified as
\begin{equation}\label{eq:Devaney V definition}
V = \left\lbrace (x,y,z) \middle|\  |x|,|y|,|z|\leq r \right\rbrace\ , ~~~~{\rm where\ }r=1+\sqrt{1+a_0}
\end{equation}
such that $f_{\rm I}(V)$ gives rise to a horseshoe. As depicted by Fig.~\ref{fig:3D_singly_folded_horseshoe}, the map expands $V$ in the unstable directions ($(x,y)$-plane), contracts it in the stable direction ($z$-axis), then folds it with a crease in the $y$-direction and stacks it along the $z$-direction. The intersection $V\cap f_{\rm I}(V)$ consists of a pair of three-dimensional ``horizontal" slabs, which, when viewed in every $(x,z)$-slice, reduces to a pair of two-dimensional horizontal strips. Furthermore, the resulting non-wandering set $\Lambda = \bigcap_{n=-\infty}^{\infty}f^n_{\rm I}(V)$ is uniformly hyperbolic and conjugate to a full shift on two symbols. 

\begin{figure}
        \centering
        \begin{subfigure}[b]{0.45\textwidth}   
            \centering 
            \includegraphics[width=\textwidth]{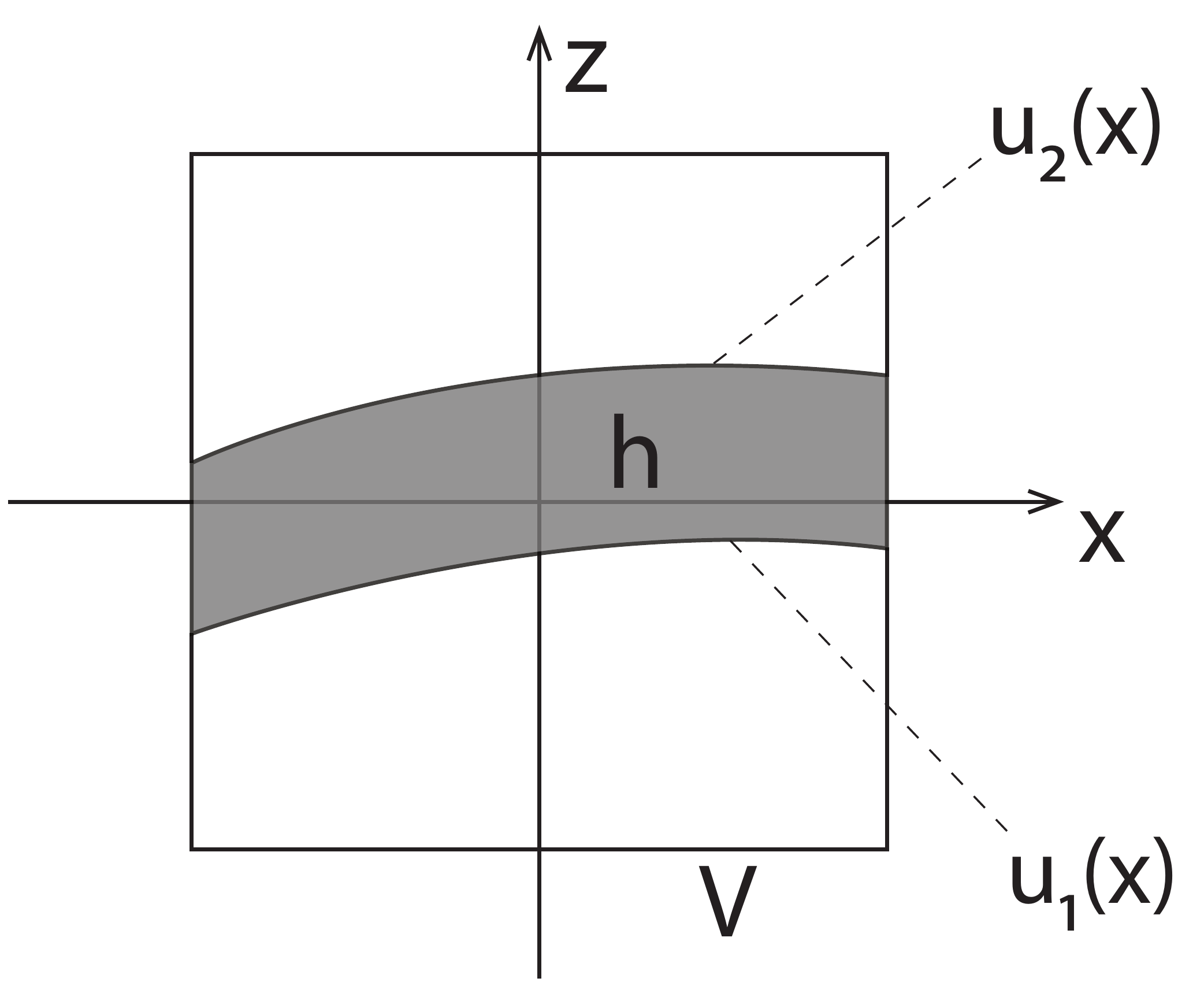}
            \caption[]%
            {}    
            \label{fig:Horizontal_strip_normal}
        \end{subfigure}
        \hfill
        \begin{subfigure}[b]{0.45\textwidth}   
            \centering 
            \includegraphics[width=\textwidth]{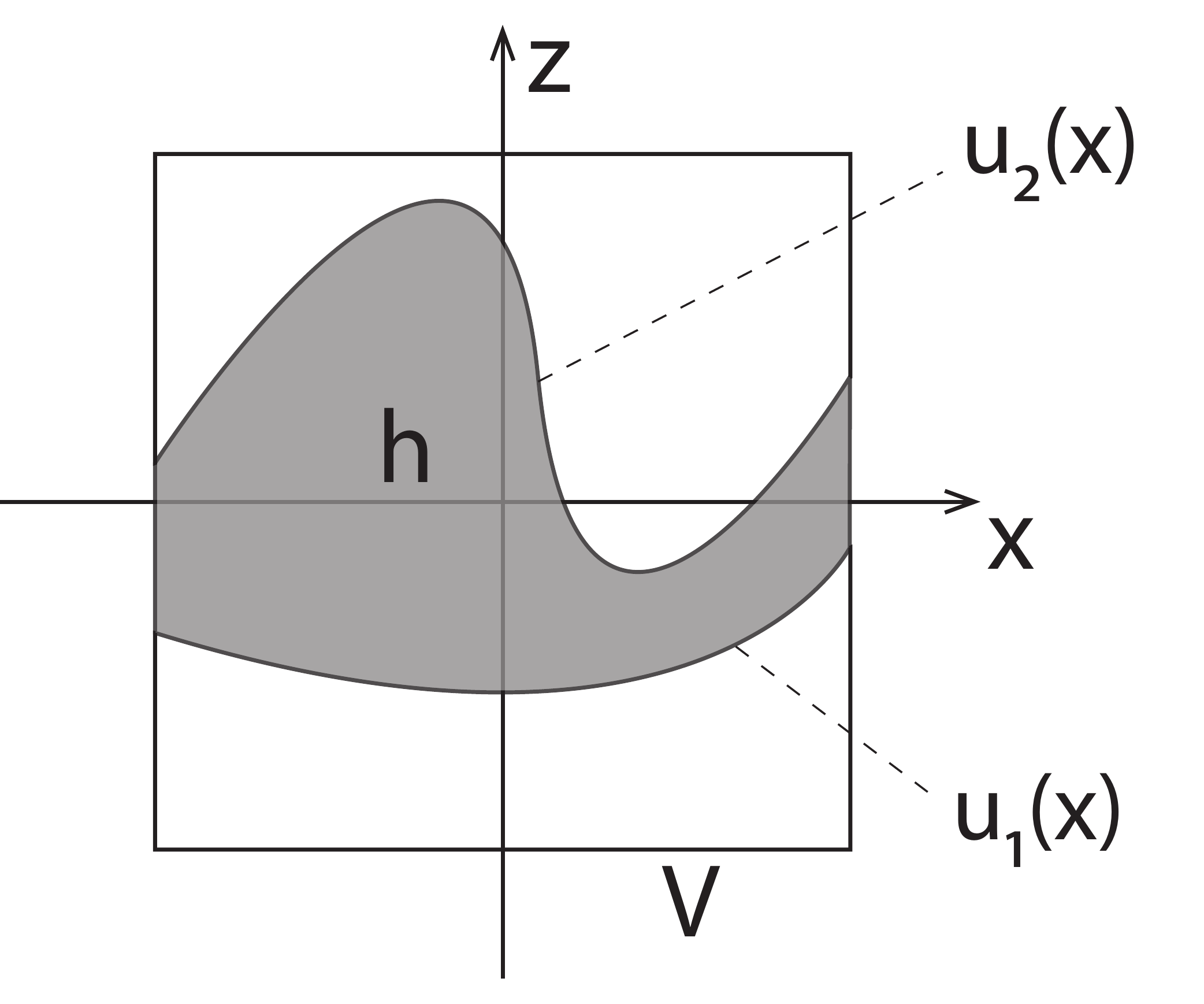}
            \caption[]%
            {}    
            \label{fig:Horizontal_strip_weird}
        \end{subfigure}
        \caption[]
        {Two examples of horizontal strip. The top and bottom boundaries of $h$ are $u_2(x)$ and $u_1(x)$, respectively, from Definition~\ref{Horizontal Strips}. The vertical boundary of $h$ is a subset of the vertical boundary of $V$. It is apparent that the strip in panel (a) is more ``well-behaved" compared to that in panel (b). In principle, certain Lipschitz conditions can be imposed to exclude the ill-behaved strips such as the one in panel (b). However, to keep the derivations simpler, we do not impose such conditions in this research. } 
        \label{fig:Horizontal_strip}
    \end{figure}

To avoid potential ambiguities, we now give precise definitions of two-dimensional horizontal strips and three-dimensional horizontal slabs.
\begin{definition}[Horizontal Strip]
\label{Horizontal Strips}
Let $V$ be a square in the $(x,z)$-plane centered at the origin:
\begin{equation}
V = \left\lbrace (x,z) \middle|\  |x|,|z|\leq r \right\rbrace\ .
\end{equation}
A set $h$ is a \textbf{horizontal strip of $V$} if there exist two curves $z=u_1(x)$ and $z=u_2(x)$ for which
\begin{equation}
-r \leq u_1(x) < u_2(x) \leq r
\end{equation}
such that 
\begin{equation}\label{eq:def horizontal strip}
h = \lbrace (x,z) | -r \leq x \leq r,\ u_1(x) \leq z \leq u_2(x) \rbrace\ .
\end{equation}
\end{definition}
\begin{remark}
This definition is quite general as it does not impose any Lipschitz conditions on the boundary curves $u_1(x)$ and $u_2(x)$. This is in contrast to the definitions of ``horizontal strips" in \cite{Moser01} and \cite{Wiggins88} (Sec. 2.3 therein), where Lipschitz conditions were imposed to guarantee that the regularity of horizontal and vertical strips so that the non-wandering set is fully conjugate to a subshift on finite symbols. Since the purpose of this article is to establish the existence of topological horseshoes, which may include the cases where the non-wandering set is only semi-conjugate to a subshift on finite symbols, to make the proceeding derivations simpler, we do not impose such Lipschitz conditions in the current study. Two essential properties of horizontal strips are that they intersect $V$ fully in the horizontal direction, i.e., no marginal intersections are allowed, and they are subsets of $V$ that partition $V$ into disjoint regions. See Fig.~\ref{fig:Horizontal_strip} for some examples. 
\end{remark}
Following \cite{Wiggins88}, the definition of horizontal strips in two dimensions can be generalized to horizontal slabs in three dimensions. 
\begin{definition}[Horizontal Slab]
\label{Horizontal Slabs}
Let $V$ be a cube in the $(x,y,z)$-space centered at the origin:
\begin{equation}
V = \left\lbrace (x,y,z) \middle|\  |x|,|y|,|z|\leq r \right\rbrace\ .
\end{equation}
A set $h$ is a \textbf{horizontal slab of $V$} if there exists two surfaces $z=u_1(x,y)$ and $z=u_2(x,y)$ for which
\begin{equation}
-r \leq u_1(x,y) < u_2(x,y) \leq r
\end{equation}
such that 
\begin{equation}\label{eq:def horizontal slabs}
h = \lbrace (x,y,z) | -r \leq x \leq r,\  -r \leq y \leq r,\ \ u_1(x,y) \leq z \leq u_2(x,y) \rbrace\ .
\end{equation}
\end{definition}
An example of a horizontal slab is given by Fig.~\ref{fig:Horizontal_slab}. 
\begin{figure}
        \centering
        \includegraphics[width=0.5\linewidth]{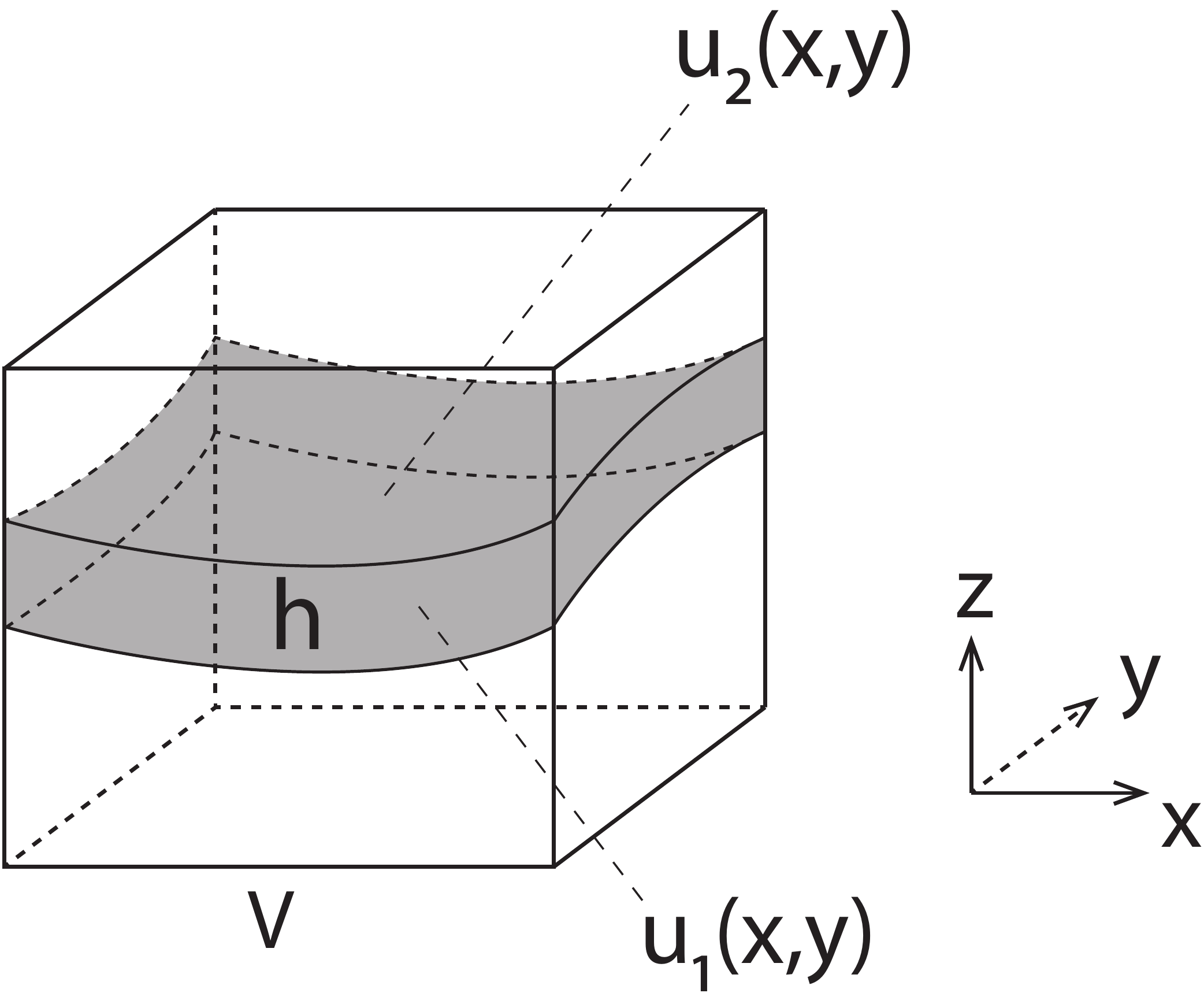} 
        \caption{$h$ is a horizontal slab of the three-dimensional cube $V$. The top and bottom boundaries of $h$ are $u_2(x,y)$ and $u_1(x,y)$, respectively, from Definition~\ref{Horizontal Slabs}. The vertical boundary of $h$ is a subset of the vertical boundary of $V$. } \label{fig:Horizontal_slab}
\end{figure}

In the next sections we will introduce horseshoes that are folded for multiple times and therefore exhibit more complicated topologies. These horseshoes are constructed by exploring various possibilities of folding and stacking in three and four dimensions.

\section{Topology II: doubly-folded horseshoe in three dimensions}\label{Type II}
The first generalization that we introduce here, namely \textit{Topology II}, is based on Figure 3.2.47 of \cite{Wiggins88}. Qualitatively, it can be considered as a Topology-I horseshoe further folded with a crease along $x$, as illustrated by Fig.~\ref{fig:3D_doubly_folded_horseshoe}. 
\begin{figure}
        \centering
        \includegraphics[width=0.7\linewidth]{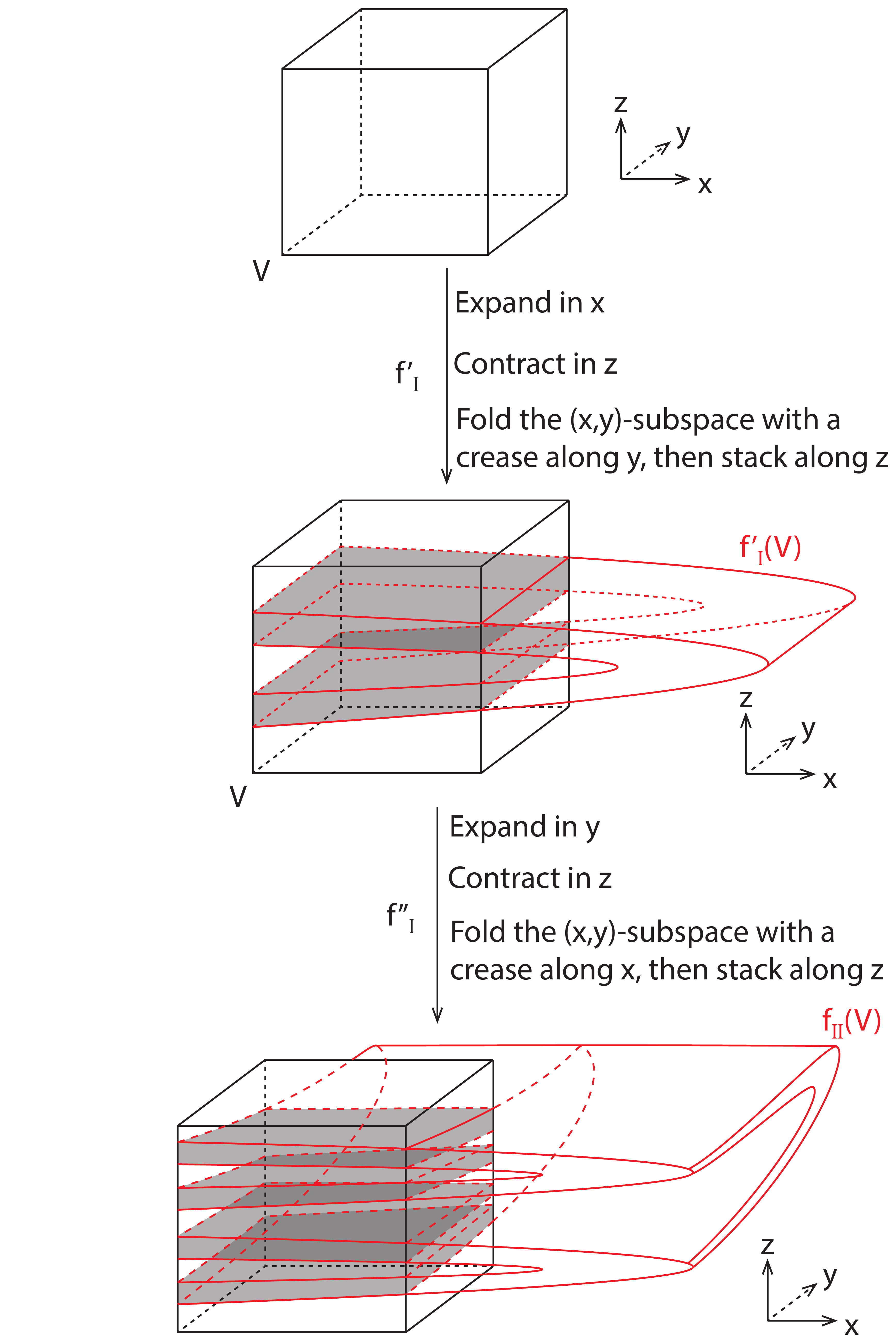} 
        \caption{(Schematic, color online) Topology II: doubly-folded horseshoe in three dimensions. Given a rectangle $V$, $f^{\prime}_{\rm I}$ expands it in $x$, contracts it in $z$, folds it with a crease along $y$ and stacks along $z$. Then subsequently, $f^{\prime\prime}_{\rm I}$ expands it in $y$, contracts it in $z$, folds it with a crease along $x$ and then stacks along $z$ again, giving rise to the doubly-folded structure. The intersection $V \cap f_{\rm II}(V)$ consists of four disjoint horizontal slabs, as shown in the bottom figure. Notice that the two foldings have independent creases (the first crease is along $y$ and the second crease is along $x$) but share the same stacking direction ($z$).    } \label{fig:3D_doubly_folded_horseshoe}
\end{figure}
A realization is also provided by the H\'{e}non-type map
\begin{equation}\label{eq:3D doubly folded Henon}
\left( \begin{array}{ccc}
x_{n+1}\\
y_{n+1} \\
z_{n+1} \end{array} \right) = 
f_{\rm II} \left( \begin{array}{ccc}
x_{n}\\
y_{n} \\
z_{n} \end{array} \right) =
\left( \begin{array}{ccc}
a_0 - x^2_n - z_n\\
a_1 - y^2_n - x_n\\
y_n \end{array} \right)
\end{equation}
with parameters $a_0,a_1 > 5+2\sqrt{5}$. The mapping $f_{\rm II}$ can be written as the compound mapping of two successive Type-I horseshoes, namely $f^{\prime}_{\rm I}$ and $f^{\prime\prime}_{\rm I}$:
\begin{equation}\label{eq:3D doubly folded Henon compound}
f_{\rm II} = f^{\prime\prime}_{\rm I} \circ f^{\prime}_{\rm I}
\end{equation}
where $f^{\prime}_{\rm I}$ takes the form Eq.~(\ref{eq:3D singly folded Henon}) with $b=1$:
\begin{equation}\label{eq:f prime I}
\left( \begin{array}{ccc}
x^{\prime}\\
y^{\prime} \\
z^{\prime} \end{array} \right) = 
f^{\prime}_{\rm I} \left( \begin{array}{ccc}
x\\
y \\
z \end{array} \right) =
\left( \begin{array}{ccc}
a_0 - x^2 - z\\
y \\
x \end{array} \right)
\end{equation}
and $f^{\prime \prime}_{\rm I}$ resembles $f^{\prime}_{\rm I}$ but interchanges the roles of the $x$ and $y$ axis, i.e., it expands $V$ in $y$, contracts $V$ in $z$, folds it with a crease along $x$, then stacks along $z$. Therefore, the mapping equations of $f^{\prime \prime}_{\rm I}$ is obtained from that of $f^{\prime}_{\rm I}$ by interchanging $x$ and $y$:
\begin{equation}\label{eq:f prime prime I}
\left( \begin{array}{ccc}
x^{\prime\prime}\\
y^{\prime\prime} \\
z^{\prime\prime} \end{array} \right) = 
f^{\prime\prime}_{\rm I} \left( \begin{array}{ccc}
x^{\prime}\\
y^{\prime} \\
z^{\prime} \end{array} \right) =
\left( \begin{array}{ccc}
x^{\prime}\\
a_1 - (y^{\prime})^2 - z^{\prime} \\
y^{\prime} \end{array} \right)\ .
\end{equation}
The inverse map $f^{-1}_{\rm II}$ is slightly complicated as it involves quartic terms:
\begin{eqnarray}
 \left( \begin{array}{ccc}
x_{n-1}\\
y_{n-1} \\
z_{n-1} \end{array} \right) & = f^{-1}_{\rm II} \left( \begin{array}{ccc}
x_{n}\\
y_{n} \\
z_{n} \end{array} \right) \nonumber \\
& =
\left( \begin{array}{ccc}
-y_n - z^2_n + a_1\\
z_n\\
-x_n - y^2_n -2 y_n z^2_n - z^4_n + a_0 + 2 a_1 y_n + 2 a_1 z^2_n - a^2_1 \end{array} \right). \label{eq:f II inverse}
\end{eqnarray}
We propose the following theorem on the topological structure of $f_{\rm II}$:
\begin{theorem}\label{3D Doubly folded horseshoe topology}
Let $a_0 = a_1 = a > 5+2\sqrt{5}$ and $r= (1+\sqrt{1+a})$. Let $V$ be a hypercube centered at the origin with side length $2r$, i.e.,
\begin{equation}\label{eq:V f II}
V = \left\lbrace (x,y,z) \Big| |x|,|y|,|z| \leq r \right\rbrace\ . 
\end{equation}  
Then the intersection $V \cap f_{\rm II}(V)$ consists of four disjoint horizontal slabs of $V$, as shown by Fig.~\ref{fig:3D_doubly_folded_horseshoe}. 
\end{theorem}
\begin{proof}
Using the identity relation $f_{\rm II}^{-1}(f_{\rm II}(V))=V$, we obtain an analytic expression for $f_{\rm II}(V)$
\begin{equation}\label{eq:f II V}
f_{\rm II}(V) = 
\left\lbrace (x,y,z) \middle| \begin{array}{ccc}
|z| \leq r \\
|-y - z^2 + a| \leq r \\
|-x - y^2 -2 y z^2 - z^4 + a + 2 a (y + z^2) - a^2| \leq r \end{array} \right\rbrace 
\end{equation}
The expression for $V \cap f_{\rm II}(V)$ is then obtained trivially by imposing the additional bounds on $x$ and $y$
\begin{equation}\label{eq:V intersect f II V}
V \cap f_{\rm II}(V) = 
\left\lbrace (x,y,z) \middle| \begin{array}{ccc}
|x|,|y|,|z| \leq r \\
|-y - z^2 + a| \leq r \\
|-x - y^2 -2 y z^2 - z^4 + a + 2 a (y + z^2) - a^2| \leq r \end{array} \right\rbrace 
\end{equation}
Let $\Sigma^2(y)$ be the $(x,z)$-plane parameterized by $y$ (the superscript indicates the dimensionality of the plane), i.e.,
\begin{equation}\label{eq: x z slice for f_II}
\Sigma^2(y) = \left\lbrace (x',y',z') \middle| x',z'\in \mathbb{R},\ y'=y \right\rbrace
\end{equation}
and similarly $\Sigma^2(x)$ the $(y,z)$-plane parameterized by $x$, i.e.,
\begin{equation}\label{eq: y z slice for f_II}
\Sigma^2(x) = \left\lbrace (x',y',z') \middle| y',z'\in \mathbb{R},\ x'=x \right\rbrace \ .
\end{equation}
Moreover, define the line segments
\begin{equation}
S^{\pm}_y = \left\lbrace (y,z) \middle| y=\pm r, |z| \leq r \right\rbrace
\end{equation}
i.e., $S^+_y$ and $S^-_y$ are the right and left boundaries of $V$, respectively, in each $\Sigma^2(x)$ slice, as labeled in Fig.~\ref{fig:3D_dooubly_folded_horseshoe_y_z_slice_Gammas}.

To prove this theorem, we need to show that $V \cap f_{\rm II}(V)$ consists of
\begin{itemize}
\item[(a)] Four disjoint horizontal strips in every $\Sigma^2(x)$ for $|x| \leq r$;
\item[(b)] Four disjoint horizontal strips in every $\Sigma^2(y)$ for $|y| \leq r$. 
\end{itemize} 
We now establish conditions (a) and (b) individually.

\begin{figure}
        \centering
        \includegraphics[width=0.6\linewidth]{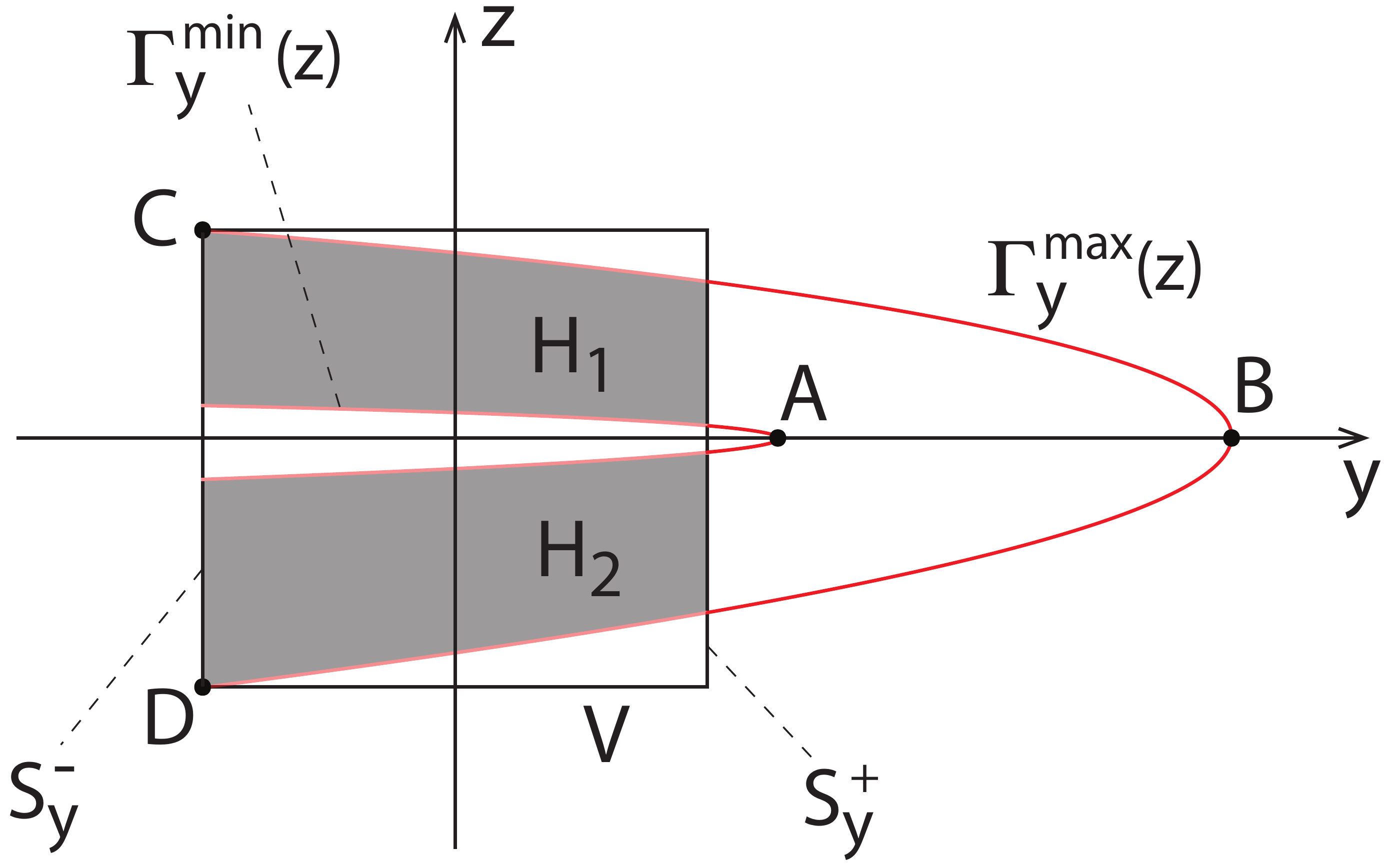} 
        \caption{(Schematic, color online) The gap bounded by $\Gamma^{\min}_y(z)$, $\Gamma^{\max}_y(z)$ , and $S^{\pm}_y$ consists of two disjoint horizontal strips $H_1$ and $H_2$.}   \label{fig:3D_dooubly_folded_horseshoe_y_z_slice_Gammas}
\end{figure}

Condition (a): the second row of Eq.~(\ref{eq:V intersect f II V}) can be rewritten in the parameterized form
\begin{equation}\label{eq:y z parabola 1 f_II}
y = -z^2 + a + s, ~~~~{\rm where}~ |s| \leq r \ .
\end{equation}
Let $\Gamma_y(z,s)$ be the family of parabolas in $\Sigma^2(x)$
\begin{equation}\label{eq:Gamma y in terms of z s f_II}
\Gamma_y(z,s) = -z^2 + a + s
\end{equation}
where $s$ is viewed as a parameter within range $|s| \leq r$. It is obvious that $\Gamma_y(z,s)$ is bounded by
\begin{equation}\label{eq:Gamma y in terms of z s bounds}
\Gamma^{\min}_y(z) \leq \Gamma_y(z,s) \leq \Gamma^{\max}_y(z)
\end{equation}
with lower and upper bounds
\begin{eqnarray}
& \Gamma^{\min}_y(z) = \Gamma_y(z,s)|_{s=-r} = -z^2 + a -r  \label{eq:Gamma y z min} \\
& \Gamma^{\max}_y(z) =  \Gamma_y(z,s)|_{s=r} = -z^2 + a + r  \label{eq:Gamma y z max}\ .
\end{eqnarray}
When viewed in each $\Sigma^2(x)$ slice, $f_{\rm II}(V)$ is the gap region between the two parabolas $\Gamma^{\min}_y(z)$ and $\Gamma^{\max}_y(z)$ (see Fig.~\ref{fig:3D_dooubly_folded_horseshoe_y_z_slice_Gammas}). The possible location of the two parabolas can be further narrowed down by establishing the following facts:
\begin{itemize}
\item[(a.1)] The vertex of $\Gamma^{\min}_y(z)$ is located on the right-hand side of $S^+_y$, as labeled by $A$ in Fig.~\ref{fig:3D_dooubly_folded_horseshoe_y_z_slice_Gammas};
\item[(a.2)] $\Gamma^{\max}_y(z)$ intersects $S^-_y$ at two points, as labeled by $C$ and $D$ in Fig.~\ref{fig:3D_dooubly_folded_horseshoe_y_z_slice_Gammas}. 
\end{itemize} 

To establish (a.1), let $A=(y_A,z_A)$. It can be solved easily that
\begin{equation}
z_A=0, ~~~~ y_A = \Gamma^{\min}_y(z_A=0)=a-r. 
\end{equation}
Using the assumption that $a>5+2\sqrt{5}$, it is straightforward to verify that $a-2r>0$, thus $y_A >r$, i.e., $A$ is on the right-hand side of $S^{+}_y$. 

To establish (a.2), notice that
\begin{equation}
\Gamma^{\max}_y(z=\pm r) = -r^2 +a +r = -r\ ,
\end{equation}
thus $C$ and $D$ are located at
\begin{equation}
C = (-r,r), ~~~~ D=(-r,-r)\ ,
\end{equation}
i.e., $C$ and $D$ are the upper-left and lower-left corners of $V$, respectively, as labeled in Fig.~\ref{fig:3D_dooubly_folded_horseshoe_y_z_slice_Gammas}. Therefore, $\Gamma^{\max}_y(z)$ intersects $S^-_y$ at its two endpoints. 

Combining (a.1) and (a.2), we know that the region bounded by $\Gamma^{\max}_y(z)$, $\Gamma^{\min}_y(z)$, and $S^{\pm}_y$ consists of two disjoint horizontal strips, labeled by $H_1$ and $H_2$ in Fig.~\ref{fig:3D_dooubly_folded_horseshoe_y_z_slice_Gammas}. Strictly speaking, both $H_1$ and $H_2$ depend on $x$, i.e., the position of $(y,z)$-slice along the $x$-axis, therefore should be written as $H_1(x)$ and $H_2(x)$. However, since the $x$-dependence will not be used for the rest of the proof, we simply omit it and write the horizontal strips without explicit $x$-dependence. When viewed in each $\Sigma^2(x)$-slice, $V \cap f_{\rm II}(V)$ can only exist inside $H_1$ and $H_2$:
\begin{equation}\label{eq:V intersects f_II V y z slice first bound}
V \cap f_{\rm II}(V) \Big|_{\Sigma^2(x)} \subset H_1 \cup H_2 \ .
\end{equation}

At this point, let us notice that Eq.~(\ref{eq:V intersects f_II V y z slice first bound}) only makes use of the second row of Eq.~(\ref{eq:V intersect f II V}), thus only provides a crude bound for $V \cap f_{\rm II}(V) \Big|_{\Sigma^2(x)}$. Based upon Eq.~(\ref{eq:V intersects f_II V y z slice first bound}), we now further refine the bound for $V \cap f_{\rm II}(V) \Big|_{\Sigma^2(x)}$ by imposing the third row of Eq.~(\ref{eq:V intersect f II V}). 

The third row of Eq.~(\ref{eq:V intersect f II V}) can be rewritten into the parameterized form
\begin{equation}\label{eq:y z parabola 2 f_II}
-x - y^2 -2 y z^2 - z^4 + a + 2 a (y + z^2) - a^2 = -s, ~~~~{\rm where}~ |s| \leq r
\end{equation}
from which we solve for $y$ and obtain two branches of solutions:
\begin{equation}\label{eq: y z parabola 2 f_II two branches}
y_{\pm}(z,x,s) = -z^2 + a \pm \sqrt{s-x+a}\ .
\end{equation}
Accordingly, let us define two families of parabolas in $\Sigma^2(x)$, denoted by $\Lambda^{\pm}_y(z,x,s)$, where
\begin{equation}\label{eq:Lambda y in terms of z x s}
\Lambda^{\pm}_y(z,x,s) = -z^2 + a \pm \sqrt{s-x+a}
\end{equation}
where $x$ and $s$ are viewed as parameters with bounds $|x|,|s| \leq r$. When viewed in each $\Sigma^2(x)$ slice, $\Lambda^{+}_y(z,x,s)$ is a family of parabolas parameterized by $s$, bounded by
\begin{equation}\label{eq:Lambda y + lower and upper bounds on each slice}
\Lambda^{+,1}_y(z,x) \leq \Lambda^{+}_y(z,x,s) \leq \Lambda^{+,2}_y(z,x)
\end{equation} 
where the lower and upper bounds are attained at
\begin{eqnarray}
& \Lambda^{+,1}_y(z,x) = \Lambda^+_y(z,x,s)|_{s=-r} = -z^2 + a + \sqrt{a-x-r}  \label{eq:Lambda y + 1} \\
& \Lambda^{+,2}_y(z,x) = \Lambda^+_y(z,x,s)|_{s=r} = -z^2 + a + \sqrt{a-x+r}  \label{eq:Lambda y + 2}\ .
\end{eqnarray}
Similarly, when viewed in each $\Sigma^2(x)$ slice, $\Lambda^{-}_y(z,x,s)$ is a family of parabolas parameterized by $s$, bounded by
\begin{equation}\label{eq:Lambda y - lower and upper bounds on each slice}
\Lambda^{-,1}_y(z,x) \leq \Lambda^{-}_y(z,x,s) \leq \Lambda^{-,2}_y(z,x)
\end{equation} 
where the lower and upper bounds are attained at
\begin{eqnarray}
& \Lambda^{-,1}_y(z,x) = \Lambda^-_y(z,x,s)|_{s=r} = -z^2 + a - \sqrt{a-x+r}  \label{eq:Lambda y - 1} \\
& \Lambda^{-,2}_y(z,x) = \Lambda^-_y(z,x,s)|_{s=-r} = -z^2 + a - \sqrt{a-x-r}  \label{eq:Lambda y - 2}\ .
\end{eqnarray}

\begin{figure}
        \centering
        \includegraphics[width=0.7\linewidth]{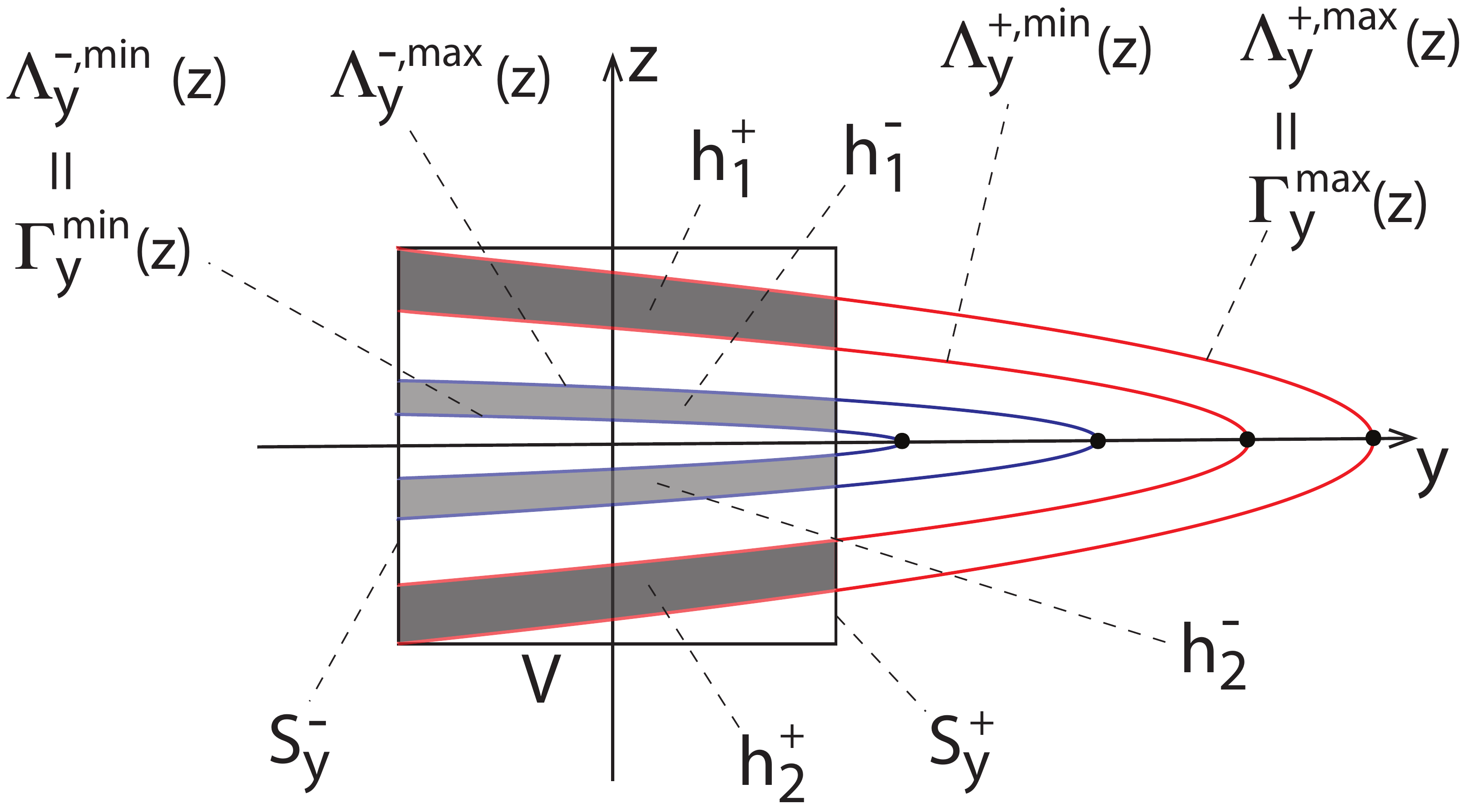} 
        \caption{(Schematic, color online) The four parabolas in Eqs.~(\ref{eq:Lambda y + min})-(\ref{eq:Lambda y - max}) cut $V$ into four disjoint horizontal strips, labeled by $h^{+}_1$, $h^{-}_1$, $h^-_2$, and $h^+_2$, respectively (from top to bottom). Notice the two pairs of identical parabolas: $\Lambda^{-,\min}_y(z)=\Gamma^{\min}_y(z)$ and $\Lambda^{+,\max}_y(z)=\Gamma^{\max}_y(z)$.}   \label{fig:3D_doubly_folded_horseshoe_y_z_slice_Lambdas}
\end{figure}

It is desirable to get rid of the $x$-dependence in Eqs.~(\ref{eq:Lambda y + lower and upper bounds on each slice}) and (\ref{eq:Lambda y - lower and upper bounds on each slice}). This can be done by obtaining uniform lower and upper bounds for $\Lambda^{\pm}_y(z,x,s)$ with respect to change in $(x,s)$. A simple calculation shows:
\begin{equation}\label{eq:Gamma y +- uniform bounds}
\Lambda^{\pm,\min}_y (z) \leq \Lambda^{\pm}_y(z,x,s) \leq \Lambda^{\pm,\max}_y (z)
\end{equation}
where the bounds are attained at
\begin{eqnarray}
& \Lambda^{+,\min}_y(z) = \Lambda^+_y(z,x,s)|_{(x,s)=(r,-r)} = -z^2 + a + \sqrt{a-2r}  \label{eq:Lambda y + min} \\
& \Lambda^{+,\max}_y(z) = \Lambda^+_y(z,x,s)|_{(x,s)=(-r,r)} = -z^2 + a + \sqrt{a+2r}  \label{eq:Lambda y + max} \\
& \Lambda^{-,\min}_y(z) = \Lambda^-_y(z,x,s)|_{(x,s)=(-r,r)} = -z^2 + a - \sqrt{a+2r}  \label{eq:Lambda y - min} \\
& \Lambda^{-,\max}_y(z) = \Lambda^-_y(z,x,s)|_{(x,s)=(r,-r)} = -z^2 + a - \sqrt{a-2r}  \label{eq:Lambda y - max}\ .
\end{eqnarray}
A schematic illustration of the four parabolas is given in Fig.~\ref{fig:3D_doubly_folded_horseshoe_y_z_slice_Lambdas}. At this point, it is worthwhile checking that since $a > 5 + 2\sqrt{5}$, we have 
\begin{equation}\label{eq:a-2r}
a-2r >0\ ,
\end{equation}
i.e., the square roots in Eqs.~(\ref{eq:Lambda y + min}) and (\ref{eq:Lambda y - max}) are real-valued. Also, it is easy to check that $r = \sqrt{a+2r}$, therefore we obtain the important relations
\begin{eqnarray}\label{eq:Identical parabola bounds}
& \Lambda^{+,\max}_y(z) = \Gamma^{\max}_y(z) \label{eq:Identical parabola bounds max} \\
& \Lambda^{-,\min}_y(z) = \Gamma^{\min}_y(z) \ , \label{eq:Identical parabola bounds min} 
\end{eqnarray}
as indicated by Fig.~\ref{fig:3D_doubly_folded_horseshoe_y_z_slice_Lambdas}. Therefore, conditions (a.1) and (a.2) immediately apply to $\Lambda^{-,\min}_y(z)$ and $\Lambda^{+,\max}_y(z)$, respectively. This guarantees that the four parabolas ($\Lambda^{\pm,\max}_y(z)$ and $\Lambda^{\pm,\min}_y(z)$) cut $V$ into four disjoint horizontal strips, as labeled by $h^{\pm}_1$ and $h^{\pm}_2$ in Fig.~\ref{fig:3D_doubly_folded_horseshoe_y_z_slice_Lambdas}. Furthermore, Eqs.~(\ref{eq:Identical parabola bounds max}) and (\ref{eq:Identical parabola bounds min}) also guarantee that $h^{\pm}_1 \subset H_1$ and $h^{\pm}_2 \subset H_2$. Hence when viewed in each $\Sigma^2(x)$ slice (see Fig.~\ref{fig:3D_doubly_folded_horseshoe_y_z_slice_Lambdas}), $V \cap f_{\rm II}(V)\Big|_{\Sigma^2(x)}$ lies within the four strips 
\begin{equation}\label{eq:V intersects f_II V y z slice second bound}
V \cap f_{\rm II}(V) \Big|_{\Sigma^2(x)} \subset h^{+}_1\cup h^{-}_1 \cup h^{+}_2 \cup h^{-}_2 \subset H_1 \cup H_2 \ ,
\end{equation}
which establishes condition (a): $V \cap f_{\rm II}(V) $ consists of four disjoint horizontal strips in every $\Sigma^2(x)$ for $|x| \leq r$. 

Condition (b): from Eq.~(\ref{eq:y z parabola 2 f_II}) we solve for $x$ and obtain
\begin{equation}\label{eq:x z parabola f_II}
x = -z^4 + 2(a-y)z^2 - (y-a)^2 + a + s, ~~~~{\rm where~} |s| \leq r.  
\end{equation}
Correspondingly, define the family of quartic polynomials
\begin{equation}\label{eq:quartic function f_II}
\Gamma_x(z,y,s) = -z^4 + 2(a-y)z^2 - (y-a)^2 + a + s\ .
\end{equation}
\begin{figure}
        \centering
        \includegraphics[width=0.8\linewidth]{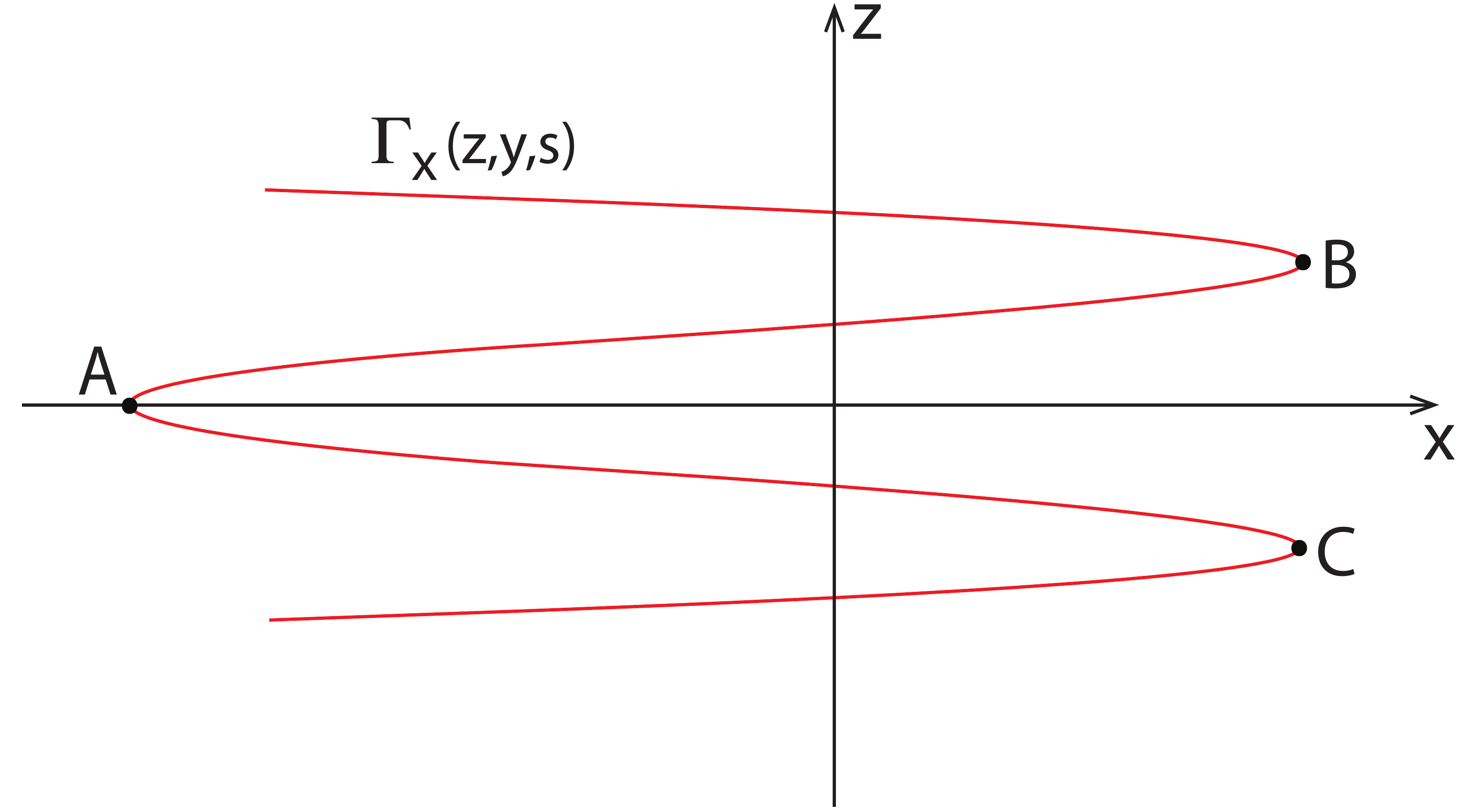} 
        \caption{(Schematic, color online) Within each $\Sigma^2(y)$ slice, $\Gamma_x(z,y,s)$ is a quartic function of $z$. It has one local minimum $A$, and two global maximums $B$ and $C$. }   \label{fig:Quartic_function_schematic}
\end{figure}
When viewed in $\Sigma^2(y)$, $\Gamma_x(z,y,s)$ is a parameterized quartic function of $z$ with parameters $y$ and $s$. By solving for ${\rm d} \Gamma_x/{\rm d}z=0$ and ${\rm d}^2 \Gamma_x/{\rm d}z^2=0$ we can easily obtain the three extremals of the quartic function, namely $A$, $B$, and $C$, where
\begin{eqnarray}
& A = (x_A,z_A) = \Big( -(y-a)^2+a+s,0 \Big) ~~{\rm is~a~local~minimum} \label{eq:extremal A f_II} \\
& B = (x_B,z_B) = \Big( a+s,\sqrt{a-y} \Big) ~~{\rm is~a~global~maximum} \label{eq:extremal B f_II} \\
& C = (x_C,z_C) = \Big( a+s,-\sqrt{a-y} \Big) ~~{\rm is~a~global~maximum} \label{eq:extremal C f_II}
\end{eqnarray}
as demonstrated schematically by Fig.~\ref{fig:Quartic_function_schematic}. Upon changing the value of $s$ within the range $|s| \leq r$, $\Gamma_x(z,y,s)$ is bounded by
\begin{equation}\label{eq:Quartic function f_II bounds}
\Gamma^{\min}_x(z,y) \leq \Gamma_x(z,y,s) \leq \Gamma^{\max}_x(z,y)
\end{equation}
where
\begin{eqnarray}
&\Gamma^{\min}_x(z,y) = \Gamma_x(z,y,s)\Big|_{s=-r} = -z^4 + 2(a-y)z^2 - (y-a)^2 + a - r \label{eq:Quartic function f_II lower bound} \\
& \Gamma^{\max}_x(z,y) = \Gamma_x(z,y,s)\Big|_{s=r} = -z^4 + 2(a-y)z^2 - (y-a)^2 + a + r\ . \label{eq:Quartic function f_II upper bound}
\end{eqnarray}
\begin{figure}
        \centering
        \includegraphics[width=0.8\linewidth]{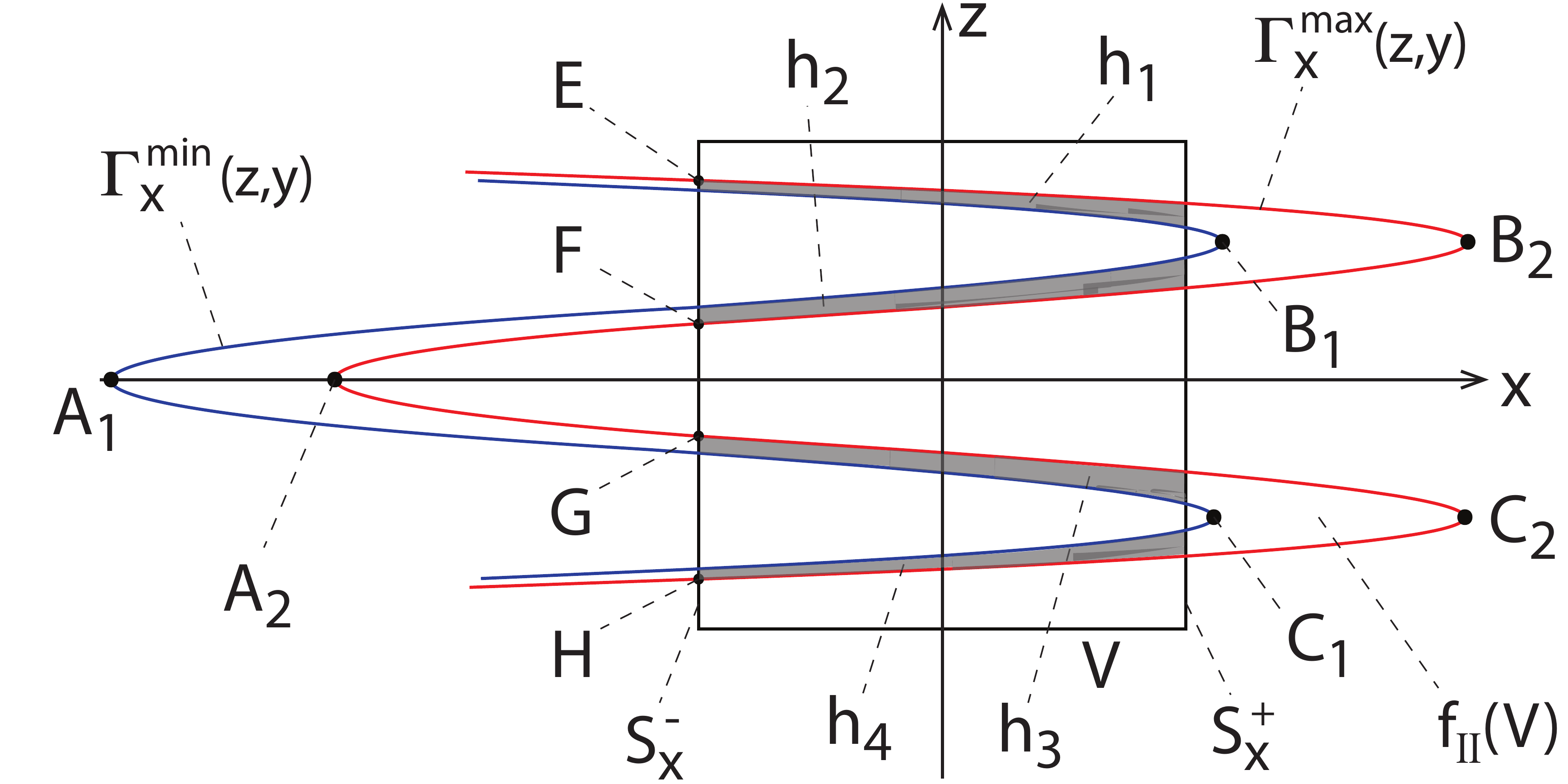} 
        \caption{(Schematic, color online) Within each $\Sigma^2(y)$ slice, $f_{\rm II}(V)$ is the region in between $\Gamma^{\min}_x(z,y)$ (blue) and $\Gamma^{\max}_x(z,y)$ (red). $V \cap f_{\rm II}(V)$ consists of four disjoint horizontal strips labeled by $h_1$, $h_2$, $h_3$, and $h_4$ (shaded). }   \label{fig:3D_doubly_folded_horseshoe_x_z_slice}
\end{figure}
Therefore within each $\Sigma^2(y)$, $f_{\rm II}(V)$ is the gap region bounded from left and right by $\Gamma^{\min}_x(z,y)$ and $\Gamma^{\max}_x(z,y)$, respectively, as shown by Fig.~\ref{fig:3D_doubly_folded_horseshoe_x_z_slice}. Let $A_1$, $B_1$, $C_1$ be the three extremals of $\Gamma^{\min}_x(z,y)$, and $A_2$, $B_2$, $C_2$ be the three extremals of $\Gamma^{\max}_x(z,y)$, as labeled in the figure. Moreover, define the line segments
\begin{equation}
S^{\pm}_x = \left\lbrace (x,z) \middle| x=\pm r, |z| \leq r \right\rbrace
\end{equation}
i.e., $S^+_x$ and $S^-_x$ are the right and left boundaries of $V$, respectively, in $\Sigma^2(y)$. The following two conditions are sufficient for condition (b):
\begin{itemize}
\item[(b.1)] The global maximums of $\Gamma^{\min}_x(z,y)$, labeled by $B_1$ and $C_1$ in Fig.~\ref{fig:3D_doubly_folded_horseshoe_x_z_slice}, are on the right-hand side of $S^+_x$;
\item[(b.2)] $\Gamma^{\max}_x(z,y)$ intersects $S^-_x$ at four points, as labeled by $E$, $F$, $G$, and $H$ in Fig.~\ref{fig:3D_doubly_folded_horseshoe_x_z_slice}. 
\end{itemize}
We now show that conditions (b.1) and (b.2) hold for all $|y| \leq r$. 

Condition (b.1): The global maximums of $\Gamma^{\min}_x(z,y)$ are located at
\begin{equation}
B_1 = (x_{B_1}, z_{B_1}), ~~~~ C_1 = (x_{C_1},z_{C_1})
\end{equation}
where $x_{B_1} = x_{C_1} = a-r >r$. Thus (b.1) is established. 

Condition (b.2): First, notice that the local minimum of $\Gamma^{\max}_x(z,y)$ is $A_2 = (x_{A_2},0) $ where
\begin{equation}\label{eq:A_2 f_II}
x_{A_2} = -(y-a)^2 +a +r \leq -(r-a)^2+a+r < -r^2 +a +r = -r
\end{equation}
where the second inequality comes from the fact that $a-2r>0$. Therefore, $A_2$ is located on the left-hand side of $S^-_x$, as shown in Fig.~\ref{fig:3D_doubly_folded_horseshoe_x_z_slice}. Second, it can be verified easily that
\begin{equation}\label{eq:E H existence f_II}
\Gamma^{\max}_x(z,y)\Big|_{z=\pm r} \leq -r.
\end{equation}
Combining Eqs.~(\ref{eq:A_2 f_II}) and (\ref{eq:E H existence f_II}), (b.2) is established as well. Therefore we have proved condition (b), i.e., $V \cap f_{\rm II}(V)\Big|_{\Sigma^2(y)}$ consists of four disjoint horizontal strips, as labeled by $h_1$, $h_2$, $h_3$, and $h_4$ in Fig.~\ref{fig:3D_doubly_folded_horseshoe_x_z_slice}. 
\end{proof}

\section{Topology III: horseshoes in four dimensions}\label{Type III}
In the preceding sections we have demonstrated two H\'{e}non-type maps in three dimensions, namely $f_{\rm I}$ and $f_{\rm II}$, where $f_{\rm I}$ folds $V$ once with a creasing along an expanding direction ($y$), and $f_{\rm II}$ folds $V$ twice with creases along independent expanding directions ($y$ and $x$). As explained in Fig.~\ref{fig:3D_doubly_folded_horseshoe}, because there are only three dimensions, the two folding operations in $f_{\rm II}$ share the same stacking direction ($z$), which results in four disjoint horizontal slabs in Fig.~\ref{fig:3D_doubly_folded_horseshoe}. A natural question is then what would happen if the dimensionality increases to four. In this section, to address this question, we consider a H\'{e}non-type map in four dimensions and introduce a new type of doubly-folded horseshoe which can exist only in dimensions $\geq 4$. Then we show that upon changing the parameters of the map, the doubly-folded horseshoe unfolds into a singly-folded horseshoe in four dimensions, which represents a reduction of topological entropy of the system. 

Consider the H\'{e}non-type map in four dimensions, namely $f_{\rm III}$
\begin{equation}\label{eq:4D doubly folded Henon}
\left( \begin{array}{ccc}
x_{n+1}\\
y_{n+1} \\
z_{n+1} \\
w_{n+1} \end{array} \right) = 
f_{\rm III} \left( \begin{array}{ccc}
x_{n}\\
y_{n} \\
z_{n}\\
w_{n} \end{array} \right) =
\left( \begin{array}{ccc}
a_0 - x^2_n - z_n + c(x_n - y_n)\\
a_1 - y^2_n - w_n - c(x_n - y_n)\\
x_n \\
y_n \end{array} \right)
\end{equation}
which is the coupled H\'{e}non map studied in Part I of this article \cite{Fujioka23}. The parameters here are $a_0$, $a_1$, and $c$. $a_0$ and $a_1$ control the rate of expansion in $x$ and $y$, respectively, and $c$ controls the coupling strength between the dynamics in the $(x,z)$-plane and the $(y,w)$-plane. 

The inverse of $f_{\rm III}$ is 
\begin{equation}\label{eq:4D doubly folded Henon inverse}
\left( \begin{array}{ccc}
x_{n-1}\\
y_{n-1} \\
z_{n-1} \\
w_{n-1} \end{array} \right) = 
f^{-1}_{\rm III} \left( \begin{array}{ccc}
x_{n}\\
y_{n} \\
z_{n}\\
w_{n} \end{array} \right) =
\left( \begin{array}{ccc}
z_n \\
w_n \\
a_0 - z^2_n - x_n + c(z_n - w_n) \\
a_1 - w^2_n - y_n - c(z_n - w_n)\end{array} \right)\ .
\end{equation}
Notice that by replacement of variables $(x,y,z,w) \mapsto (z,w,x,y)$, $f_{\rm III}$ is transformed into $f^{-1}_{\rm III}$. 

In Section. 2.2 of Part I, we have proposed two types of Anti-Integrable (AI) limits \cite{Aubry90} of Eq.~(\ref{eq:4D doubly folded Henon}), namely Type A and Type B, where Type A is an AI limit with four symbols (Eq.~(2.7) therein) and Type B is an AI limit with two symbols (Eq.~(2.9) therein). Type A is obtained by taking $a_0=a_1\to\infty$ while keeping $c$ fixed and finite. Intuitively speaking, it gives rise to infinite expansion rates within the $(x,z)$ and $(y,w)$ planes but only allows finite coupling between the two planes. Type B is obtained by taking $a_0=a_1\to\infty$ while keeping $c/\sqrt{a}=\gamma$ constant. Intuitively, it also gives rise to infinite expansion rates within the $(x,z)$ and $(y,w)$ planes but imposes a coupling strength $c$ proportional to $\sqrt{a}$. As we show next, the topologies of the horseshoes near these two AI limits are fundamentally different: Type A is a doubly-folded horseshoe topologically equivalent to a direct product between a pair of two-dimensional singly-folded horseshoes in the $(x,z)$ and $(y,w)$ planes, while Type B is a singly-folded horseshoe in four dimensions. Therefore, when changing the values of parameters $(a_0,a_1,c)$ from the neighborhood of Type A to the neighborhood of Type B, global bifurcations must happen that unfold the doubly-folded horseshoe into the singly-folded horseshoe. 
 
\subsection{Topology III-A: doubly folded horseshoe in four dimensions, with independent stacking directions}\label{Type III-A}

Topology III-A can only exist in dimensions $\geq 4$. Using our example here, it can be realized near the Type-A AI limit of Eq.~(\ref{eq:4D doubly folded Henon}). To better understand the Type-A AI limit, a simple case to start with is when $c=0$, i.e., zero coupling, for which $f_{\rm III}$ reduces to a direct-product between the two-dimensional H\'{e}non maps in the $(x,z)$-plane and the $(y,w)$-plane. Therefore when $c=0$ and $a_0$, $a_1$ values greater than the bound given by Devaney and Nitecki ($a_0$, $a_1 > 5+2\sqrt{5}$), $f_{\rm III}$ is identical to a direct-product between the Smale horseshoe maps in the $(x,z)$-plane and the $(y,w)$-plane, as illustrated by  Fig.~\ref{fig:4D_doubly_folded_horseshoe}.
\begin{figure}
        \centering
        \includegraphics[width=0.7\linewidth]{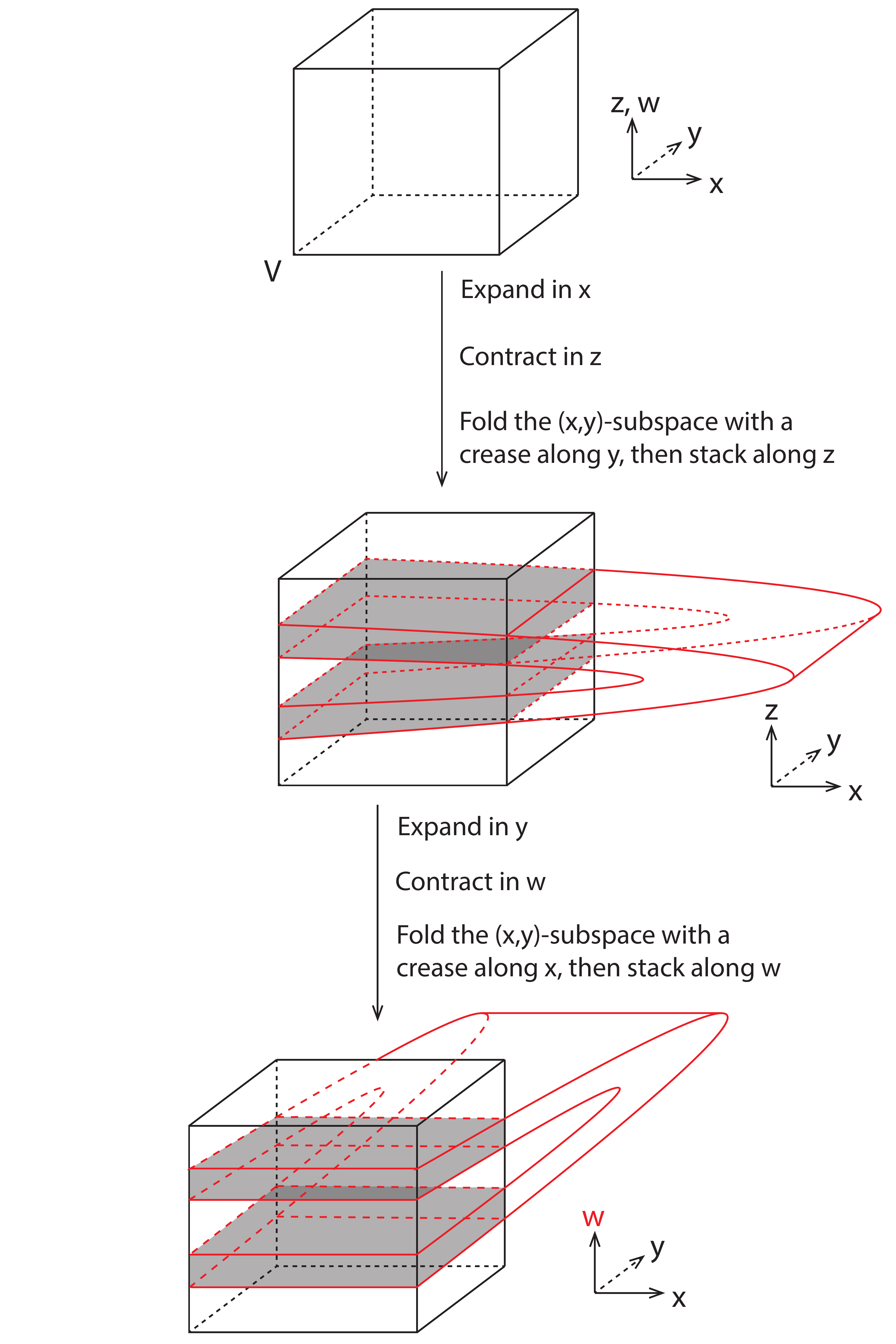} 
        \caption{(Schematic, color online) Topology III-A: doubly-folded horseshoe in four dimensions. There are two folding operations. The first operation folds the $(x,y)$-subspace with a crease along $y$ and stacks it along $z$. The second operation folds the $(x,y)$-subspace with a crease along $x$ and stacks it along $w$.  } \label{fig:4D_doubly_folded_horseshoe}
\end{figure}

Geometrically, the action of $f_{\rm III}$ on $V$ involves two steps: the first step expands $V$ in $x$, contracts it in $z$, folds it with a crease along $y$, and stacks it along $z$; the second step expands it in $y$, contracts it in $w$, folds it with a crease along $x$, and stacks it along $w$. Although $f_{\rm III}$ and $f_{\rm II}$ both involve two folding operations, the difference between them is critical: the two folding operations in $f_{\rm II}$ share the same stacking direction ($z$), while the two folding operations in $f_{\rm III}$ have independent stacking directions ($z$ and $w$). Therefore, unlike $V\cap f_{\rm II}(V)$ which is composed by four disjoint horizontal slabs, $V \cap f_{\rm III}(V)$ consists of four hypercylinders, each being a direct product of a horizontal strip in the $(x,z)$-plane and a horizontal strip in the $(y,w)$-plane. Equivalently speaking, when $c=0$, $V \cap f_{\rm III}(V)$ is a direct product between two horizontal strips in the $(x,z)$-plane and two horizontal strips in the $(y,w)$-plane, as demonstrated by Fig.~\ref{fig:4D_doubly_folded_horseshoe}.

When $c>0$, it is reasonable to expect that if $c$ is finite and $a_0,a_1\gg c>0$, i.e., when the expansion rates in the $(x,z)$ and $(y,w)$ planes (governed by $a_0$ and $a_1$, respectively) are much greater compared to the coupling strength $c$ between the two planes, the coupling could be neglected and the resulting topology of $V \cap f_{\rm III}(V)$ should be identical to the $c=0$ case. Therefore when $a_0,a_1\gg c>0$, we expect $V \cap f_{\rm III}(V)$ to be topologically equivalent to a direct product of a pair of two-dimensional Smale horseshoes, as illustrated by Fig.~\ref{fig:4D_doubly_folded_horseshoe}. We now give a concrete proof of this under some analytical bounds on the parameters $a_0$, $a_1$, and $c$. 

\begin{theorem}\label{4D Doubly folded horseshoe topology}
Let $A_0 = (a_0 + a_1)/2 \geq -1$ and $r= 2 \sqrt{2}(1+\sqrt{1+A_0})$. Let $V$ be a hypercube centered at the origin with side length $2r$, i.e.,
\begin{equation}\label{eq:V}
V = \left\lbrace (x,y,z,w) \Big| |x|,|y|,|z|,|w| \leq r \right\rbrace\ . 
\end{equation}  
Given the bounds on parameters
\begin{eqnarray}
& 0 < \frac{1}{4}c^2 + a_i - (c+2) r\ ,\ (i=0,1) \label{eq:4D Doubly folded horseshoe topology parameter bound 1} \\
& 0 \leq r^2 - 2(c+1) r - a_i\ ,\ (i=0,1) \label{eq:4D Doubly folded horseshoe topology parameter bound 2}
\end{eqnarray}
the intersection $V \cap f_{\rm III}(V)$ is homeomorphic to a direct product between two disjoint horizontal strips in the $(x,z)$-plane and two disjoint horizontal strips in the $(y,w)$-plane, i.e., it retains the direct-product structure shown by Fig.~\ref{fig:4D_doubly_folded_horseshoe}. 
\end{theorem}
\begin{remark}
Note that the definitions of $A_0$ and $r$ are the same as in Part I of this article \cite{Fujioka23}, and the bounds provided by Eqs.~(\ref{eq:4D Doubly folded horseshoe topology parameter bound 1}) and (\ref{eq:4D Doubly folded horseshoe topology parameter bound 2}) are identical to Eqs. (3.12) and (3.13) of Part I as well. 
\end{remark}
\begin{proof}
Let $\Sigma^2(y,w)$ be the $(x,z)$-plane parameterized by $(y,w)$, i.e.,
\begin{equation}\label{eq: x z slice}
\Sigma^2(y,w) = \left\lbrace (x',y',z',w')\ \middle|\ x',z'\in \mathbb{R},\ y'=y,\ w'=w \right\rbrace
\end{equation}
and similarly $\Sigma^2(x,z)$ the $(y,w)$-plane parameterized by $(x,z)$, i.e.,
\begin{equation}\label{eq: y w slice}
\Sigma^2(x,z) = \left\lbrace (x',y',z',w')\ \middle|\ y',w'\in \mathbb{R},\ x'=x,\ z'=z \right\rbrace \ .
\end{equation}
Moreover, let $V \cap f_{\rm III}(V)|_{ \Sigma^2(y,w)}$ denote the restriction of $V \cap f_{\rm III}(V)\cap \Sigma^2(y,w)$ on $\Sigma^2(y,w)$, i.e., the $(x,z)$-slice of $V \cap f_{\rm III}(V)$. Similarly, let $V \cap f_{\rm III}(V)|_{ \Sigma^2(x,z)}$ denote the restriction of $V\cap f_{\rm III}(V)\cap \Sigma^2(x,z)$ on $\Sigma^2(x,z)$, i.e., the $(y,w)$-slice of $V\cap f_{\rm III}(V)$.

Using the identity relation $f_{\rm III}^{-1}(f_{\rm III}(V))=V$ we obtain an analytic expression for $f_{\rm III}(V)$
\begin{equation}\label{eq:f III V}
f_{\rm III}(V) = 
\left\lbrace (x,y,z,w) \middle| \begin{array}{ccc}
|z| \leq r \\
|w| \leq r \\
|a_0 - z^2 - x + c(z-w)| \leq r \\
|a_1 - w^2 - y - c(z-w)| \leq r \end{array} \right\rbrace \ .
\end{equation}

The expression for $V \cap f_{\rm III}(V)$ is then easily obtained by imposing the additional constraints of $|x|,|y| \leq r$:
\begin{equation}\label{eq:V intersect f III V}
V \cap f_{\rm III}(V) = 
\left\lbrace (x,y,z,w) \middle| \begin{array}{ccc}
|x|,|y|,|z|,|w| \leq r \\
|a_0 - z^2 - x + c(z-w)| \leq r \\
|a_1 - w^2 - y - c(z-w)| \leq r \end{array} \right\rbrace \ .
\end{equation}

To prove this theorem, it is sufficient to establish the following two conditions:
\begin{itemize}
\item[(a)] $V \cap f_{\rm III}(V)|_{ \Sigma^2(y,w)}$ consists of two disjoint horizontal strips for all $|y|,|w| \leq r$.
\item[(b)] $V \cap f_{\rm III}(V)|_{ \Sigma^2(x,z)}$ consists of two disjoint horizontal strips for all $|x|,|z| \leq r$.
\end{itemize}
We now establish these two conditions individually. 

Condition (a): the second row of Eq.~(\ref{eq:V intersect f III V}) is equivalent to 
\begin{equation}\label{eq:x z parameter s}
-s = a_0 - z^2 - x + c(z-w), \ \textrm{where }|s|\leq r
\end{equation}
and thus 
\begin{equation}\label{eq:x z equation f III}
x=-z^2 + cz + a_0 + s-cw, \ \textrm{where }|s|\leq r\ .
\end{equation}
Correspondingly, define a family of parabolas 
\begin{equation}\label{eq:Gamma x in terms of z w s}
\Gamma_x(z,w,s)=-z^2 + cz + a_0 + s-cw\ ,
\end{equation}
which is viewed as a quadratic function of $z$ with parameters $w$ and $s$ bounded by $|w|,|s|\leq r$.
 
With the help of Eqs.~(\ref{eq:V intersect f III V}) and (\ref{eq:Gamma x in terms of z w s}), $V \cap f_{\rm III}(V)|_{ \Sigma^2(y,w)}$ (for which $|y|,|w| \leq r$) can be expressed as
\begin{equation}\label{eq:V intersect f III V x z slice}
V \cap f_{\rm III}(V)\Big|_{ \Sigma^2(y,w)} = \left\lbrace (x,z) \middle| \begin{array}{ccc}
|x|,|z|,|s| \leq r \\
x=\Gamma_x(z,w,s) \end{array} \right\rbrace \ ,
\end{equation}
as illustrated schematically by Fig.~\ref{fig:4D_doubly_folded_horseshoe_x_z}.

\begin{figure}
        \centering
        \includegraphics[width=0.6\linewidth]{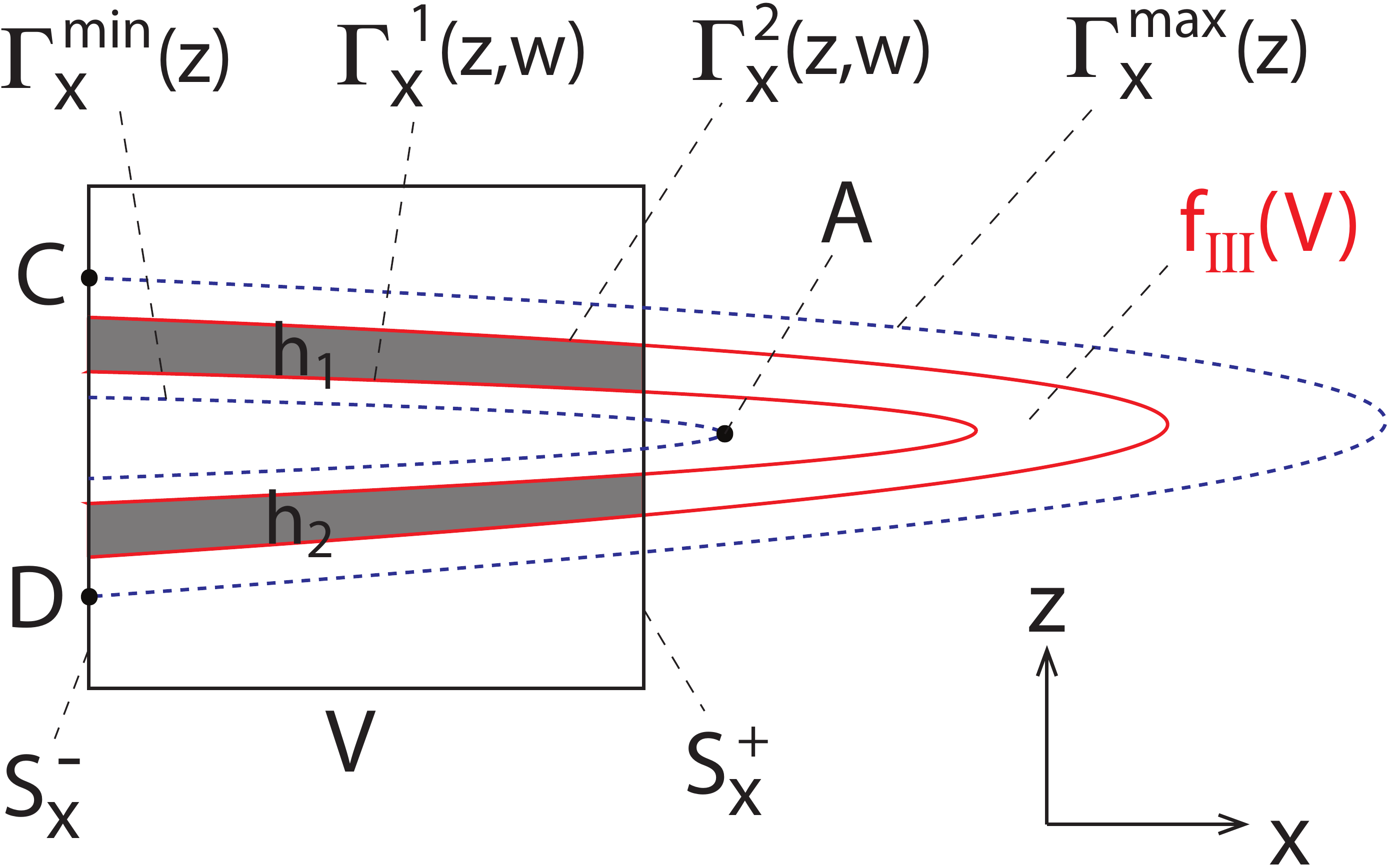} 
        \caption{(Schematic, color online) Visualization of $V \cap f_{\rm III}(V)|_{ \Sigma^2(y,w)}$, which consists of two disjoint horizontal strips $h_1$ and $h_2$.  } \label{fig:4D_doubly_folded_horseshoe_x_z}
\end{figure}

Let $\Gamma^1_x(z,w)$ and $\Gamma^2_x(z,w)$ be parabolas in $\Sigma^2(y,w)$ obtained by setting the $s$ parameter in Eq.~(\ref{eq:Gamma x in terms of z w s}) to $-r$ and $r$, respectively: 
\begin{eqnarray}
& \Gamma^1_x(z,w) =\Gamma_x(z,w,s)|_{s=-r} = -z^2 +cz + a_0 -r -cw  \label{eq:Gamma x 1} \\
& \Gamma^2_x(z,w) =\Gamma_x(z,w,s)|_{s=r} = -z^2 +cz + a_0 +r -cw  \label{eq:Gamma x 2}\ .
\end{eqnarray}
Furthermore, let Let $\Gamma^{\min}_x(z)$ and $\Gamma^{\max}_x(z)$ be parabolas in $\Sigma^2(y,w)$ obtained by setting $(s,w)=(-r,r)$ and $(s,w)=(r,-r)$, respectively:
\begin{eqnarray}
& \Gamma^{\min}_x(z) = \Gamma_x(z,w,s)|_{(s,w)=(-r,r)} = -z^2 +cz + a_0 - (c+1)r  \label{eq:Gamma x min} \\
& \Gamma^{\max}_x(z) = \Gamma_x(z,w,s)|_{(s,w)=(r,-r)} = -z^2 +cz + a_0 + (c+1)r  \label{eq:Gamma x max}\ .
\end{eqnarray}
Notice that the $w$-dependence is removed in the expressions of $\Gamma^{\min}_x$ and $\Gamma^{\max}_x$. Since $|w|\leq r$, $\Gamma^{\min}_x(z)$ and $\Gamma^{\max}_x(z)$ provide uniform lower and upper bounds for $\Gamma^1_x(z,w)$ and $\Gamma^2_x(z,w)$:
\begin{equation}\label{eq:Uniform bound for parabolas f III V x z}
\Gamma^{\min}_x(z) \leq \Gamma^1_x(z,w) \leq \Gamma_x(z,w,s) \leq \Gamma^2_x(z,w) \leq \Gamma^{\max}_x(z)\ .
\end{equation}
Geometrically, $\Gamma^1_x(z,w)$ and $\Gamma^2_x(z,w)$ are the left and right boundaries of $f_{\rm III}(V)$, and they must exist in the gap between $\Gamma^{\min}_x(z)$ and $\Gamma^{\max}_x(z)$, as shown by Fig.~\ref{fig:4D_doubly_folded_horseshoe_x_z}. 

Let $S^+_x$ and $S^-_x$ be the right and left boundaries of $V$, respectively
\begin{eqnarray}
& S^+_x = \left\lbrace (x,z) \middle| x=r, |z| \leq r  \right\rbrace \label{eq:S x +} \\
& S^-_x = \left\lbrace (x,z) \middle| x=-r, |z| \leq r  \right\rbrace  \label{eq:S x -} \ .
\end{eqnarray}

It is straightforward that condition (a) is equivalent to the following two conditions:
\begin{enumerate}
\item[(a.1)] The vertex of $\Gamma^{\min}_x(z)$, as labeled by $A$ in Fig.~\ref{fig:4D_doubly_folded_horseshoe_x_z}, is located on the right-hand side of $S^+_x$.
\item[(a.2)] $\Gamma^{\max}_x(z)$ intersects $S^-_x$ at two points (e.g., points $C$ and $D$ in Fig.~\ref{fig:4D_doubly_folded_horseshoe_x_z})\ .
\end{enumerate}
When both are satisfied, $f_{\rm III}(V)$ must intersect $V$ at two disjoint horizontal strips, therefore give rise two a topological binary horseshoe in every $\Sigma^2(y,w)$ slice. We now show that conditions (a.1) and (a.2) are satisfied given the bounds of Eqs.~(\ref{eq:4D Doubly folded horseshoe topology parameter bound 1}) and (\ref{eq:4D Doubly folded horseshoe topology parameter bound 2}). 

Condition (a.1): the vertex $A$ of $\Gamma^{\min}_x(z)$ is computed as
\begin{equation}\label{eq:Vertex Gamma x min}
(x_A,z_A) = \left( \frac{c^2}{4} + a_0 - (c+1)r, \frac{c}{2} \right)\ .
\end{equation}
From Eq.~(\ref{eq:4D Doubly folded horseshoe topology parameter bound 1}) we immediately know that $x_A >r$, i.e., $A$ is located on the right-hand side of $S^+_x$. 

Condition (a.2): this is equivalent to the condition that $\Gamma^{\max}_x(z=\pm r) \leq r$, which can be easily deduced from Eq.~(\ref{eq:4D Doubly folded horseshoe topology parameter bound 2}).

Therefore, Eqs.~(\ref{eq:4D Doubly folded horseshoe topology parameter bound 1}) and (\ref{eq:4D Doubly folded horseshoe topology parameter bound 2}) are sufficient for the existence of two disjoint horizontal strips, which gives rise to a topological binary horseshoe in every $\Sigma^2(y,w)$ slice. Condition (a) is established. 

Condition (b): due to the symmetry of Eq.~(\ref{eq:4D doubly folded Henon}), by interchanging $(x,z)$ with $(y,w)$, our proof for condition (a) immediately applies to condition (b). Consequently, $V \cap f_{\rm III}(V)$ is topologically equivalent to a direct product between a pair of two-dimensional Smale horseshoes, which is a doubly-folded horseshoe in four dimensions as shown by Fig.~\ref{fig:4D_doubly_folded_horseshoe}. 
\end{proof}

\subsection{Topology III-B: singly folded horseshoe in four dimensions}\label{Type III-B}

By taking the limit $a_0=a_1=a \to \infty$ while keeping $c/\sqrt{a}=\gamma=\textrm{constant}$, we obtain the Type-B AI limit. Equivalently speaking, Type-B can be attained by allowing the coupling strength $c$ in Type-A to approach infinity while keeping it in constant proportion to $\sqrt{a}$. As we will show next, the topology of Type-B AI limit is simpler compared to Type-A as it involves only one folding.

As shown by Part I of this paper, Type-B AI limit is more conveniently studied by performing the change of coordinates 
\begin{eqnarray}
	\left(\begin{array}{c}
	X\\
	Y\\
	Z\\
	W
	\end{array}\right)
	=\frac{1}{2}\left(\begin{array}{c}
	x+y\\
	x-y\\
	z+w\\
	z-w\\
	\end{array}\right), 
	\label{eq:change_of_coordinate}
\end{eqnarray}
under which Eq.~(\ref{eq:4D doubly folded Henon}) can be rewritten as
\begin{eqnarray}
\label{eq:map_tranformed}
	\left(\begin{array}{c}
    X_{n+1}\\
    Y_{n+1}\\
    Z_{n+1}\\
    W_{n+1}
  \end{array}\right)
  =F\left(\begin{array}{c}
    X_n\\
    Y_n\\
    Z_n\\
    W_n
  \end{array}\right)
  =\left(\begin{array}{c}
    A_0-(X_n^2+Y_n^2)-Z_n\\
    A_1-2X_nY_n-W_n+2cY_n\\
    X_n\\
    Y_n
  \end{array}\right)
\end{eqnarray}
where $$\displaystyle A_0=\frac{a_0+a_1}{2}, ~~~~~A_1=\frac{a_0-a_1}{2}. $$
The inverse map $F^{-1}$ is given by
\begin{eqnarray}
\label{eq:map_tranformed_inverse}
  \left(\begin{array}{c}
    X_{n-1}\\
    Y_{n-1}\\
    Z_{n-1}\\
    W_{n-1}
  \end{array}\right)
  =F^{-1}\left(\begin{array}{c}
    X_n\\
    Y_n\\
    Z_n\\
    W_n
  \end{array}\right)
  =\left(\begin{array}{c}
    Z_n\\
    W_n\\
    A_0-(Z_n^2+W_n^2)-X_n\\
   A_1-2Z_nW_n-Y_n+2cW_n\\
  \end{array}\right) .
\end{eqnarray}

\begin{theorem}\label{4D Singly folded horseshoe topology}
Let $R= 1 + \sqrt{1+ A_0}$, and $V_F$ be a hypercube centered at the origin with side length $2R$, i.e., \begin{equation}\label{eq:V F}
V_F = \left\lbrace (X,Y,Z,W) \Big| |X|,|Y|,|Z|,|W| \leq R \right\rbrace\ . 
\end{equation}  
If the parameters satisfy the following conditions:
		\begin{eqnarray}
		  \label{eq:typeB_sufficient1}
			&A_1 \le R<c,\\
		  \label{eq:typeB_sufficient2}
			&R<A_0-(W^*)^2-R,\\
		  \label{eq:typeB_sufficient3}
			&W^*\le R
		\end{eqnarray}
where $W^*=\max\Bigl(\displaystyle \Bigl|\frac{2R-A_1}{2(c-R)}\Bigr|$, $\Bigl|\displaystyle\frac{-2R-A_1}{2(c-R)}\Bigr|\Bigr)$, then $V_F \cap F(V_F)$ is homeomorphic to a direct product between two disjoint horizontal strips in the $(X,Z)$ plane and a single horizontal strip in the $(Y,W)$ plane.
\end{theorem}
\begin{proof}
Let $\Sigma^2(Y,W)$ be the $(X,Z)$-plane parameterized by $(Y,W)$:
\begin{equation}\label{eq: X Z slice}
\Sigma^2(Y,W) = \left\lbrace (X',Y',Z',W')\ \middle|\ X',Z'\in \mathbb{R},\ Y'=Y,\ W'=W \right\rbrace
\end{equation}
and similarly $\Sigma^2(X,Z)$ the $(Y,W)$-plane parameterized by $(X,Z)$:
\begin{equation}\label{eq: Y W slice}
\Sigma^2(X,Z) = \left\lbrace (X',Y',Z',W')\ \middle|\ Y',W'\in \mathbb{R},\ X'=X,\ Z'=Z \right\rbrace \ .
\end{equation}
Moreover, let $V_F \cap F(V_F)|_{ \Sigma^2(Y,W)}$ denote the restriction of $V_F \cap F(V_F)\cap \Sigma^2(Y,W)$ on $\Sigma^2(Y,W)$, i.e., the $(X,Z)$-slice of $V_F \cap F(V_F)$. Similarly, let $V_F \cap F(V_F)|_{ \Sigma^2(X,Z)}$ denote the restriction of $V_F\cap F(V_F)\cap \Sigma^2(X,Z)$ on $\Sigma^2(X,Z)$, i.e., the $(Y,W)$-slice of $V_F\cap F(V_F)$.

Using the identity relation $F^{-1}(F(V_F)) = V_F$, we obtain an analytic expression for $F(V_F)$
\begin{equation}\label{eq:F V_F}
F(V_F) = 
\left\lbrace (X,Y,Z,W) \middle| \begin{array}{ccc}
|Z|,|W| \leq R \\
|A_0 - (Z^2+W^2) - X | \leq R \\
|A_1 - 2ZW - Y + 2cW| \leq R \end{array} \right\rbrace \ ,
\end{equation}
therefore the expression for $V_F \cap F(V_F)$ is obtained by imposing the trivial constraints on $X$ and $Y$:
\begin{equation}\label{eq:V_F intersect F V_F}
V_F \cap F(V_F) = 
\left\lbrace (X,Y,Z,W) \middle| \begin{array}{ccc}
|X|,|Y|,|Z|,|W| \leq R \\
|A_0 - (Z^2+W^2) - X | \leq R \\
|A_1 - 2ZW - Y + 2cW| \leq R \end{array} \right\rbrace \ .
\end{equation}

To prove this theorem, it is sufficient to establish the following two conditions:
\begin{itemize}
\item[(a)]  $V_F \cap F(V_F)|_{ \Sigma^2(X,Z)}$ consists of a single horizontal strip for all $|X|,|Z|\leq R$.
\item[(b)] $V_F \cap F(V_F)|_{ \Sigma^2(Y,W)}$ consists of two disjoint horizontal strips for all $|Y|,|W|\leq R$.
\end{itemize}
We now proceed to establish conditions (a)  and (b) individually. 

Condition (a): by restricting Eq.~(\ref{eq:V_F intersect F V_F}) to a particular $\Sigma^2(X,Z)$ plane for which $|X|,|Z| \leq R$, we obtain the expression for $V_F \cap F(V_F)|_{ \Sigma^2(X,Z)}$:
\begin{equation}\label{eq:Y W slice V_F intersect F(V_F) preliminary}
V_F \cap F(V_F) \Big|_{ \Sigma^2(X,Z)} = 
\left\lbrace (Y,W) \middle| \begin{array}{ccc}
|Y|,|W| \leq R \\
|A_1 - 2ZW - Y + 2cW| \leq R \end{array} \right\rbrace \ .
\end{equation}
The second row of Eq.~(\ref{eq:Y W slice V_F intersect F(V_F) preliminary}) is equivalent to 
\begin{equation}\label{eq:Y W relation}
Y = 2(c-Z)W + A_1 + s
\end{equation}
where $s$ is a parameter within bound $|s| \leq R$. Thus Eq.~(\ref{eq:Y W slice V_F intersect F(V_F) preliminary}) can be rewritten as
\begin{equation}\label{eq:Y W slice V_F intersect F(V_F)}
V_F \cap F(V_F) \Big|_{ \Sigma^2(X,Z)} = 
\left\lbrace (Y,W) \middle| \begin{array}{ccc}
|Y|,|W|,|s| \leq R \\
Y = 2(c-Z)W + A_1 + s \end{array} \right\rbrace \ .
\end{equation}
From Eq.~(\ref{eq:Y W slice V_F intersect F(V_F)}) we solve for $W$:
\begin{equation}\label{eq:W in terms of Y Z s}
W = \Gamma_W(Y,Z,s) \equiv \frac{Y - A_1 - s}{2(c-Z)}\ .
\end{equation}
Within each $\Sigma^2(X,Z)$ plane, $\Gamma_W(Y,Z,s)$ is a function of $Y$ with fixed parameter $Z$ and varying parameter $s$. From Eq.~(\ref{eq:typeB_sufficient1}) we know $c>R$, since $|Z| \leq R$, it is obvious that the denominator in Eq.~(\ref{eq:W in terms of Y Z s}) is always positive. Therefore $\Gamma_W$ is a straight-line parameterized by $s$ with positive slope, with uniform lower and upper bounds $\Gamma_W^1 (Y,Z)$ and $\Gamma_W^2 (Y,Z)$, respectively, where
\begin{eqnarray}
&&\Gamma_W^1 (Y,Z) = \Gamma_W(Y,Z,s)|_{s=R} = \frac{Y - A_1 - R}{2(c-Z)} \\
&&\Gamma_W^2 (Y,Z) = \Gamma_W(Y,Z,s)|_{s=-R} = \frac{Y - A_1 + R}{2(c-Z)} \ . 
\end{eqnarray} 
Furthermore, the maximum and minimum of $W$ under all possible $(Y,Z,s)$ values, denoted by $W^{\max}$ and $W^{\min}$ respectively, are
\begin{eqnarray}
&&W^{\rm max} = \Gamma_W(Y,Z,s)|_{(Y,Z,s)=(R,R,-R)} = \frac{2R-A_1}{2(c-R)}\\
&&W^{\rm min} = \Gamma_W(Y,Z,s)|_{(Y,Z,s)=(-R,R,R)} =\frac{-2R-A_1}{2(c-R)}\ .
\end{eqnarray} 
$W^{\min}$ and $W^{\max}$ provide uniform upper and lower bounds for $\Gamma_W^1 (Y,Z)$ and $\Gamma_W^2 (Y,Z)$, respectively, in every $\Sigma^2(X,Z)$ plane:
\begin{equation}\label{eq:Y W uniform bounds}
W^{\min} \leq \Gamma_W^1 (Y,Z) < \Gamma_W^2 (Y,Z) \leq W^{\max},
\end{equation}
as depicted in Fig.~\ref{fig:4D_singly_folded_horseshoe_Y_W}.  
\begin{figure}
        \centering
        \includegraphics[width=0.6\linewidth]{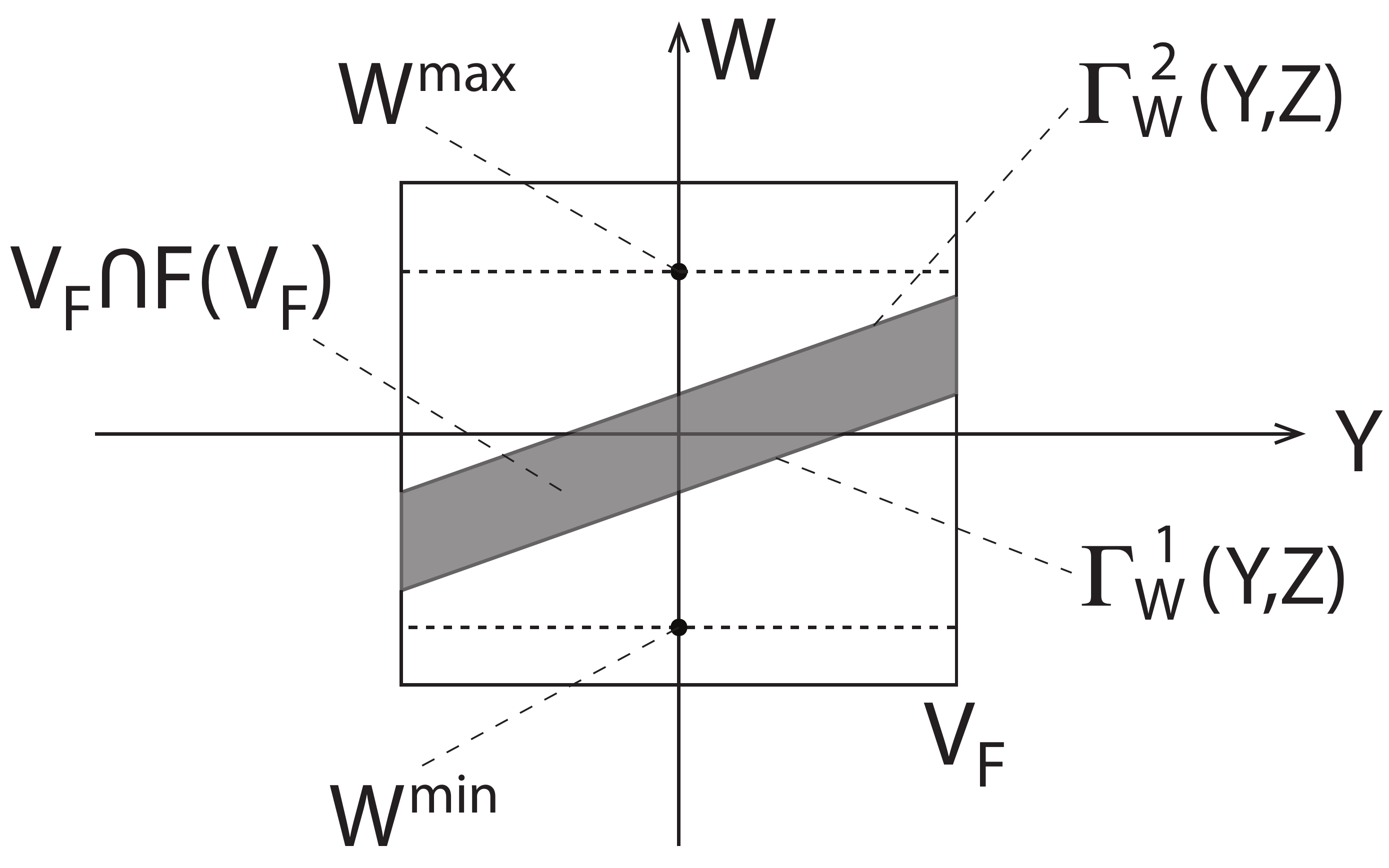} 
        \caption{(Schematic) When viewed in every $\Sigma^2(X,Z)$-slice, $V_F \cap F(V_F)$ is the gap between $\Gamma_W^1(Y,Z)$ and $\Gamma_W^2(Y,Z)$. Furthermore, $\Gamma_W^2(Y,Z)$ is bounded from above by $W^{\max}\leq W^* \leq R$, and $\Gamma_W^1(Y,Z)$ is bounded from below by $W^{\min} \geq -W^* \geq -R$. Therefore, $V_F \cap F(V_F)$ must be a horizontal strip which intersects $V_F$ fully in the $Y$-direction, and the width of $V_F \cap F(V_F)$ must be strictly smaller than $2R$ when measured in the $W$-direction.   } \label{fig:4D_singly_folded_horseshoe_Y_W}
\end{figure}
Since Eq.~(\ref{eq:typeB_sufficient3}) already guarantees that $W^{*}=\max \left( |W^{\min}|,|W^{\max}| \right) \leq R$, Eq.~(\ref{eq:Y W uniform bounds}) can be further bounded by
\begin{equation}\label{eq:Y W uniform bounds further}
-R \leq -W^{*} \leq W^{\min} \leq \Gamma_W^1 (Y,Z) < \Gamma_W^2 (Y,Z) \leq W^{\max} \leq W^{*} \leq R \ .
\end{equation}
Since $V_F \cap F(V_F)|_{\Sigma^2(X,Z)}$ is the gap between $\Gamma_W^1(Y,Z)$ and $\Gamma_W^2(Y,Z)$, Eq.~(\ref{eq:Y W uniform bounds further}) indicates that this gap is a horizontal strip which intersects $V_F$ fully in the $Y$-direction, and the width of $V_F \cap F(V_F)$ must be strictly smaller than $2R$ when measured in the $W$-direction, as shown by Fig.~\ref{fig:4D_singly_folded_horseshoe_Y_W}. Thus condition (a) is established.  

Condition (b): by constraining Eq.~(\ref{eq:V_F intersect F V_F}) to a particular $\Sigma^2(Y,W)$ plane for which $|Y|,|W| \leq R$, we obtain the expression for $V_F \cap F(V_F)|_{ \Sigma^2(Y,W)}$:
\begin{equation}\label{eq:X Z slice V_F intersect F(V_F) preliminary}
V_F \cap F(V_F) \Big|_{ \Sigma^2(Y,W)} = 
\left\lbrace (X,Z) \middle| \begin{array}{ccc}
|X|,|Z| \leq R \\
|A_0 -(Z^2+W^2) -X | \leq R \end{array} \right\rbrace \ .
\end{equation}
The second row of Eq.~(\ref{eq:X Z slice V_F intersect F(V_F) preliminary}) is equivalent to 
\begin{equation}\label{eq:X Z relation}
X = \Gamma_X(Z,W,s) \equiv -Z^2 - W^2 + A_0 + s\ .
\end{equation}
Within each particular $\Sigma^2(Y,W)$ plane, $\Gamma_X(Z,W,s)$ is a quadratic function of $Z$ with fixed parameter $W$ and varying parameter $s$ ($|s| \leq R$). Eq.~(\ref{eq:X Z slice V_F intersect F(V_F) preliminary}) can thus be rewritten as
\begin{equation}\label{eq:X Z slice V_F intersect F(V_F)}
V_F \cap F(V_F) \Big|_{ \Sigma^2(Y,W)} = 
\left\lbrace (X,Z) \middle| \begin{array}{ccc}
|X|,|Z|,|s| \leq R \\
X=\Gamma_X(Z,W,s)  \end{array} \right\rbrace \ .
\end{equation}

\begin{figure}
        \centering
        \includegraphics[width=0.6\linewidth]{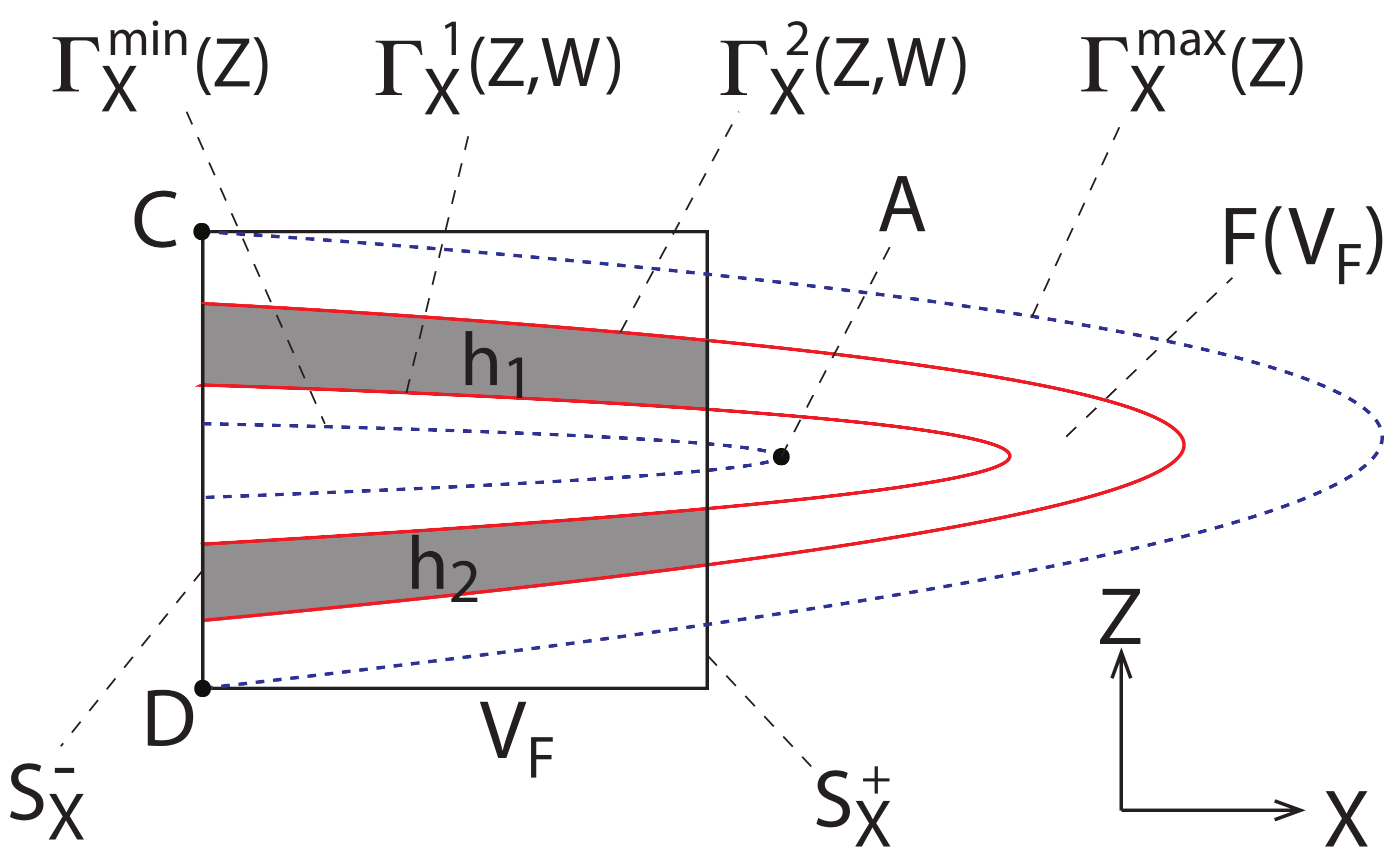} 
        \caption{(Schematic, color online) Within each $\Sigma^2(Y,W)$ plane, $V_F \cap F(V_F)|_{\Sigma^2(Y,W)}$ consists of two horizontal strips $h_1$ and $h_2$ (bounded by the two red parabolas). Therefore, $V_F \cap F(V_F)$ gives rise to a binary topological horseshoe in every $\Sigma^2(Y,W)$. } \label{fig:4D_singly_folded_horseshoe_X_Z}
\end{figure}

Let $\Gamma^1_X(Z,W)$ and $\Gamma^2_X(Z,W)$ be parabolas in $\Sigma^2(Y,W)$ obtained by setting the $s$ parameter in Eq.~(\ref{eq:X Z slice V_F intersect F(V_F)}) to $-R$ and $R$, respectively: 
\begin{eqnarray}
& \Gamma^1_X(Z,W) = \Gamma_X(Z,W,s)|_{s=-R} = -Z^2 - W^2 + A_0 -R\label{eq:Gamma X 1} \\
& \Gamma^2_X(Z,W) = \Gamma_X(Z,W,s)|_{s=R} = -Z^2 -W^2 + A_0 +R  \label{eq:Gamma X 2}\ .
\end{eqnarray}

Furthermore, let $\Gamma^{\min}_X(Z)$ and $\Gamma^{\max}_X(Z)$ be parabolas obtained by setting $(W,s)=(W^{*},-R)$ and $(W,s)=(0,R)$, respectively:
\begin{eqnarray}
& \Gamma^{\min}_X(Z) = \Gamma_X(Z,W,s)|_{(W,s)=(W^{*},-R)} = -Z^2 - (W^{*})^2 + A_0 -R  \label{eq:Gamma X min} \\
& \Gamma^{\max}_X(Z) = \Gamma_X(Z,W,s)|_{(W,s)=(0,R)} = -Z^2 + A_0 + R  \label{eq:Gamma X max}\ .
\end{eqnarray}

Notice that the $W$-dependence is removed in the expressions for $\Gamma^{\min}_X$ and $\Gamma^{\max}_X$. From Eq.~(\ref{eq:Y W uniform bounds further}) we know $|W| \leq W^{*}$, thus $\Gamma^{\min}_X(Z)$ and $\Gamma^{\max}_X(Z)$ provide uniform lower and upper bounds for $\Gamma^1_X(Z,W)$ and $\Gamma^2_X(Z,W)$:
\begin{equation}\label{eq:Uniform bound for parabolas F V_F X Z}
\Gamma^{\min}_X(Z) \leq \Gamma^1_X(Z,W)  < \Gamma^2_X(Z,W) \leq \Gamma^{\max}_X(Z)\ .
\end{equation}
Geometrically, when observed within each $\Sigma^2(Y,W)$ plane, $\Gamma^1_X(Z,W)$ and $\Gamma^2_X(Z,W)$ are the left and right boundaries of $F(V_F)$, and they must exist in the gap region bounded by $\Gamma^{\min}_X(Z)$ and $\Gamma^{\max}_X(Z)$, as shown by Fig.~\ref{fig:4D_singly_folded_horseshoe_X_Z}. 

Let $S^+_X$ and $S^-_X$ be the right and left boundaries of $V_F|_{\Sigma^2(Y,W)}$, respectively
\begin{eqnarray}
& S^+_X = \left\lbrace (X,Z) \middle| X=R, |Z| \leq R  \right\rbrace \label{eq:S X +} \\
& S^-_X = \left\lbrace (X,Z) \middle| X=-R, |Z| \leq R  \right\rbrace  \label{eq:S X -} \ .
\end{eqnarray}

To prove condition (b), it is sufficient to prove the following two conditions:
\begin{enumerate}
\item[(b.1)] The vertex of $\Gamma^{\min}_X(Z)$, labeled by $A$ in Fig.~\ref{fig:4D_singly_folded_horseshoe_X_Z}, is located on the right-hand side of $S^+_X$;
\item[(b.2)] $\Gamma^{\max}_X(Z)$ intersects $S^-_X$ at two points (points $C$ and $D$ in Fig.~\ref{fig:4D_singly_folded_horseshoe_X_Z})\ .
\end{enumerate}
We now show that conditions (b.1) and (b.2) are satisfied by the parameter bounds in Eqs.~(\ref{eq:typeB_sufficient1})-(\ref{eq:typeB_sufficient3}). 

Condition (b.1): the vertex $A$ of $\Gamma^{\min}_X(Z)$ is computed as
\begin{equation}\label{eq:Vertex Gamma X min}
(X_A,Z_A) = \left( A_0 - (W^{*})^2 -R, 0 \right)\ .
\end{equation}
From Eq.~(\ref{eq:typeB_sufficient2}) we immediately know that $X_A > R$, i.e., $A$ is located on the right-hand side of $S^+_X$. 

Condition (b.2): this is equivalent to the condition that $\Gamma^{\max}_X(Z=\pm R) \leq R$. Since
\begin{equation}\label{eq:Gamma X max S X - intersections}
\Gamma^{\max}_X(Z=\pm R) = -R^2 + A_0 +R = -R \leq -R
\end{equation}
where the second equality is due to the fact that $R=1+\sqrt{1+A_0}$, condition (b.2) follows. In fact, Eq.~(\ref{eq:Gamma X max S X - intersections}) indicates that the parabola $\Gamma^{\max}_X(Z)$ intersects $V_F$ at its two corner points, as labeled by $C$ and $D$ in Fig.~\ref{fig:4D_singly_folded_horseshoe_X_Z}.  

Therefore, Eqs.~(\ref{eq:typeB_sufficient1})-(\ref{eq:typeB_sufficient3}) are sufficient for the existence of two disjoint horizontal strips, which gives rise to topological binary horseshoes in every $\Sigma^2(Y,W)$ slice. Condition (b) is established. 
\end{proof}

\section{Topology IV: doubly-folded horseshoe in four dimensions, with a common stacking direction}\label{Type IV}
In Section.~\ref{Type II}, we have introduced Topology II, a doubly-folded horseshoe with a common stacking direction in three dimensions, with an example realization given by $f_{\rm II}$.  In this section, we introduce \textit{Topology IV}, a four-dimensional horseshoe which is topologically equivalent to a direct product between a three-dimensional Topology II horseshoe and an one-dimensional uniform contraction. 

\begin{figure}
        \centering
        \begin{subfigure}[b]{0.75\textwidth}   
            \centering 
            \includegraphics[width=\textwidth]{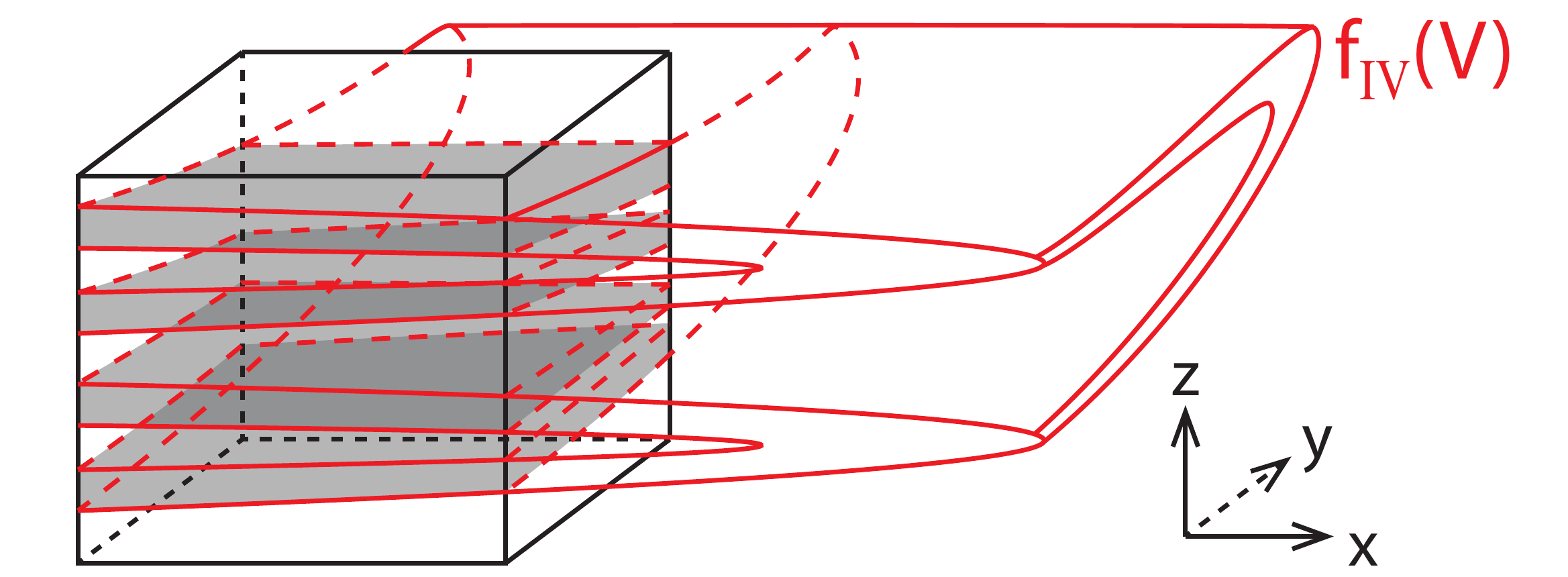}
            \caption[]%
            {$(x,y,z)$-slice}    
            \label{fig:f_IV_uncoupled_x_y_z_slice}
        \end{subfigure}
        \hfill
        \begin{subfigure}[b]{0.2\textwidth}   
            \centering 
            \includegraphics[width=\textwidth]{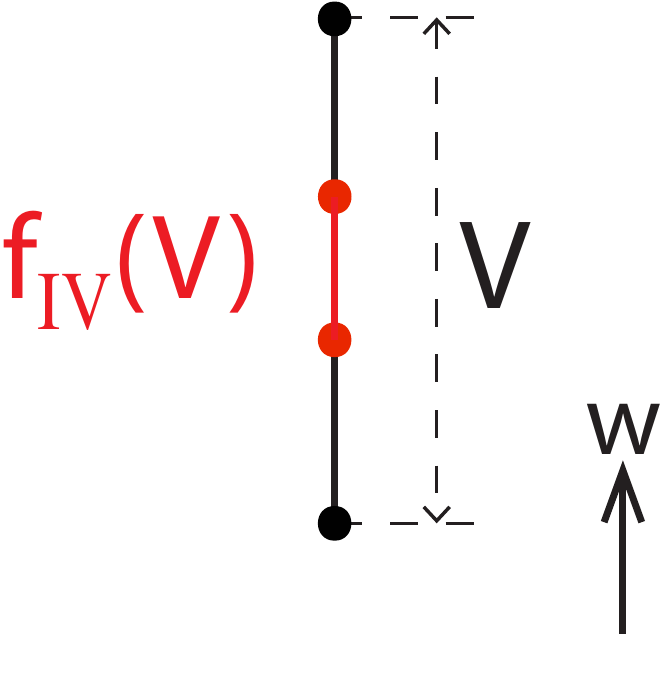}
            \caption[]%
            {$w$-slice}    
            \label{fig:f_IV_uncoupled_w_slice}
        \end{subfigure}
        \caption[]
        {(Schematic, color online) $V \cap f_{\rm IV}(V)$ when $c=0$, viewed in the $(x,y,z)$-slice (a) and the $w$-slice (b). Panel (a): $f_{\rm IV}$ is identical to $f_{\rm II}$ in every $(x,y,z)$-slice, thus giving rise to a doubly-folded horseshoe with a common stacking direction. Panel (b): in every $w$-slice, $f_{\rm IV}$ is a uniform contraction by constant factor $b$, thus $V$ is contracted to the sub-segment $f_{\rm IV}(V)$. } 
        \label{fig:f_IV_uncoupled}
    \end{figure}

The example realization that we use here is $f_{\rm IV}$, a four-dimensional generalization of $f_{\rm II}$: 
\begin{equation}\label{eq:f_IV}
\left( \begin{array}{ccc}
x^{\prime}\\
y^{\prime} \\
z^{\prime} \\
w^{\prime} \end{array} \right) = 
f_{\rm IV} \left( \begin{array}{ccc}
x\\
y \\
z \\ 
w\end{array} \right) =
\left( \begin{array}{ccc}
a_0 - x^2 - z + c(x-w)\\
a_1 - y^2 - x\\
y \\
bw + c(w-x) \end{array} \right)
\end{equation}
where $a_0$, $a_1$, $b$, and $c$ are parameters and $0\leq c<b<1$. Since we only focus on the topological structure and are not interested in the metric properties, we assume $a_0=a_1=a$ hereafter. The inverse mapping $f^{-1}_{\rm IV}$ is given by 
\begin{equation}\label{eq:f_IV inverse}
\left( \begin{array}{ccc}
x\\
y \\
z \\
w \end{array} \right) = 
f^{-1}_{\rm IV} \left( \begin{array}{ccc}
x^{\prime}\\
y^{\prime} \\
z^{\prime} \\ 
w^{\prime} \end{array} \right) =
\left( \begin{array}{ccc}
-y^{\prime}-(z^{\prime})^2+a\\
z^{\prime}\\
z(x^{\prime},y^{\prime},z^{\prime},w^{\prime}) \\
w(y^{\prime},z^{\prime},w^{\prime}) \end{array} \right)
\end{equation}
where 
\begin{eqnarray}\label{eq:f_IV inverse z component}
& z(x^{\prime},y^{\prime},z^{\prime},w^{\prime})=-(z^{\prime})^4 + \left[ 2(a-y^{\prime})-\frac{bc}{b+c} \right](z^{\prime})^2 \nonumber \\
&~~~ + \left[ -(y^{\prime})^2 + \left( 2a-\frac{bc}{b+c} \right)y^{\prime} +a -a^2 + \frac{abc}{b+c} - x^{\prime} - \frac{c}{b+c}w^{\prime} \right] 
\end{eqnarray}
and
\begin{equation}\label{eq:f_IV inverse w component}
w(y^{\prime},z^{\prime},w^{\prime})=\frac{w^{\prime}-cy^{\prime}-c(z^{\prime})^2+ca}{b+c}\ .
\end{equation}

When $c=0$, $f_{\rm IV}$ reduces to a pseudo three-dimensional map for which the dynamics in the $(x,y,z)$-hyperplane is decoupled from the dynamics along the $w$-axis. The $(x,y,z)$-dynamics is identical to $f_{\rm II}$, and $w$-dynamics is a uniform contraction by constant factor $b$. Therefore when $c=0$, under suitable choices of $a_0$, $a_1$, and $V$, we expect $V \cap f_{\rm IV}(V)$ to reproduce the topology of $V \cap f_{\rm II}(V)$ in the $(x,y,z)$-hyperplane, while in the $w$-direction, $f_{\rm IV}(V)$ is simply a sub-segment of the line segment $V$. This geometry is demonstrated by Fig.~\ref{fig:f_IV_uncoupled}. 

When $c>0$, the $c(x-w)$ and $c(w-x)$ terms in Eq.~(\ref{eq:f_IV}) introduce linear coupling between the $(x,y,z)$-subspace and the $w$-subspace dynamics, thus gives rise to dynamics more complicated than the previous $c=0$ case. However, it is reasonable to expect that when $c$ is small enough, the coupling terms should act as a perturbation to the uncoupled dynamics, thus the topological structure of $V \cap f_{\rm IV}(V)$ should be preserved. The next theorem shows that this is indeed the case for certain range of parameters. 

\begin{theorem}\label{4D Doubly folded horseshoe with cusp topology}
Let $a_0=a_1=a>32+8\sqrt{2}$, $0 \leq c < b < 1$, and $r= 2 \sqrt{2}(1+\sqrt{1+a})$. Let $V$ be a hypercube centered at the origin with side length $2r$, i.e.,
\begin{equation}\label{eq:V}
V = \left\lbrace (x,y,z,w) \Big| |x|,|y|,|z|,|w| \leq r \right\rbrace\ . 
\end{equation}  
Given the following additional bounds on parameters $b$ and $c$:
\begin{eqnarray}
& c \leq \frac{(1-b)r}{r^2+2r-a}\label{eq:4D Doubly folded horseshoe with cusp topology parameter bound 1} \\
& \frac{bc}{b+c}<2(a-r) \label{eq:4D Doubly folded horseshoe with cusp topology parameter bound 2} \\
& \frac{c}{b+c} < \frac{2r-2\sqrt{r^2-b^2(a-2r)}}{b^2}  \label{eq:4D Doubly folded horseshoe with cusp topology parameter bound 3} \\
& \frac{c}{b+c} < \frac{a(a-2r)}{ab+(1-b)r}  \label{eq:4D Doubly folded horseshoe with cusp topology parameter bound 4} \\
& \frac{c}{b+c} < \frac{r^2-2r-a}{(1+b)r}\ , \label{eq:4D Doubly folded horseshoe with cusp topology parameter bound 5}
\end{eqnarray}
the intersection $V \cap f_{\rm IV}(V)$ is homeomorphic to a direct product between four disjoint horizontal slabs in the $(x,y,z)$-subspace and an one-dimensional line segment in the $(w)$-subspace, as illustrated by Fig.~\ref{fig:f_IV_uncoupled}. 
\end{theorem}
\begin{remark}
Eqs.~(\ref{eq:4D Doubly folded horseshoe with cusp topology parameter bound 1})-(\ref{eq:4D Doubly folded horseshoe with cusp topology parameter bound 5}) provide upper bounds on $c$ such that the topology of the $c=0$ case (Fig.~\ref{fig:f_IV_uncoupled}) is preserved.    
\end{remark}
\begin{proof}
Let $\Sigma^3(w)$ be the $(x,y,z)$-hyperplane parameterized by $w$, i.e.,
\begin{equation}\label{eq: x y z slice for f_IV}
\Sigma^3(w) = \left\lbrace (\tilde{x},\tilde{y},\tilde{z},\tilde{w}) \middle|  \tilde{x},\tilde{y},\tilde{z} \in \mathbb{R},~~\tilde{w}=w \right\rbrace
\end{equation}
and $\Sigma^1(x,y,z)$ the $w$-line parameterized by $(x,y,z)$, i.e.,
\begin{equation}\label{eq: w slice for f_IV}
\Sigma^1(x,y,z) = \left\lbrace (\tilde{x},\tilde{y},\tilde{z},\tilde{w}) \middle| \tilde{w} \in \mathbb{R},~~\tilde{x}=x,~\tilde{y}=y,~\tilde{z}=z \right\rbrace \ .
\end{equation}
From the identity relation $f^{-1}_{\rm IV}\left( f_{\rm IV}(V) \right)=V$, the analytic expression for $V \cap f_{\rm IV}(V)$ can be obtained:
\begin{equation}\label{eq:V intersect f IV V}
V \cap f_{\rm IV}(V) = 
\left\lbrace (x^{\prime},y^{\prime},z^{\prime},w^{\prime}) \middle| \begin{array}{ccc}
|x^{\prime}|,|y^{\prime}|,|z^{\prime}|,|w^{\prime}| \leq r \\
|-y^{\prime}-(z^{\prime})^2+a| \leq r \\
|z(x^{\prime},y^{\prime},z^{\prime},w^{\prime})| \leq r \\
|w(y^{\prime},z^{\prime},w^{\prime})| \leq r \end{array} \right\rbrace
\end{equation}
where the functions $z(x^{\prime},y^{\prime},z^{\prime},w^{\prime})$ and $w(y^{\prime},z^{\prime},w^{\prime})$ are defined by Eq.~(\ref{eq:f_IV inverse z component}) and Eq.~(\ref{eq:f_IV inverse w component}), respectively. To prove the theorem, it is sufficient to prove that $V \cap f_{\rm IV}(V)$ is topologically equivalent to:
\begin{itemize}
\item[(a)] Topology-II horseshoe in every $\Sigma^3(w)$ for which $|w| \leq r$, as illustrated by Fig.~\ref{fig:f_IV_uncoupled_x_y_z_slice}.
\item[(b)] Line segment which is a subset of $[-r,r]$ in every $\Sigma^1(x,y,z)$ for which $|x|,|y|,|z| \leq r$, as illustrated by Fig.~\ref{fig:f_IV_uncoupled_w_slice}\ .
\end{itemize}
\begin{figure}
        \centering
        \begin{subfigure}[b]{0.35\textwidth}
            \centering
            \includegraphics[width=\textwidth]{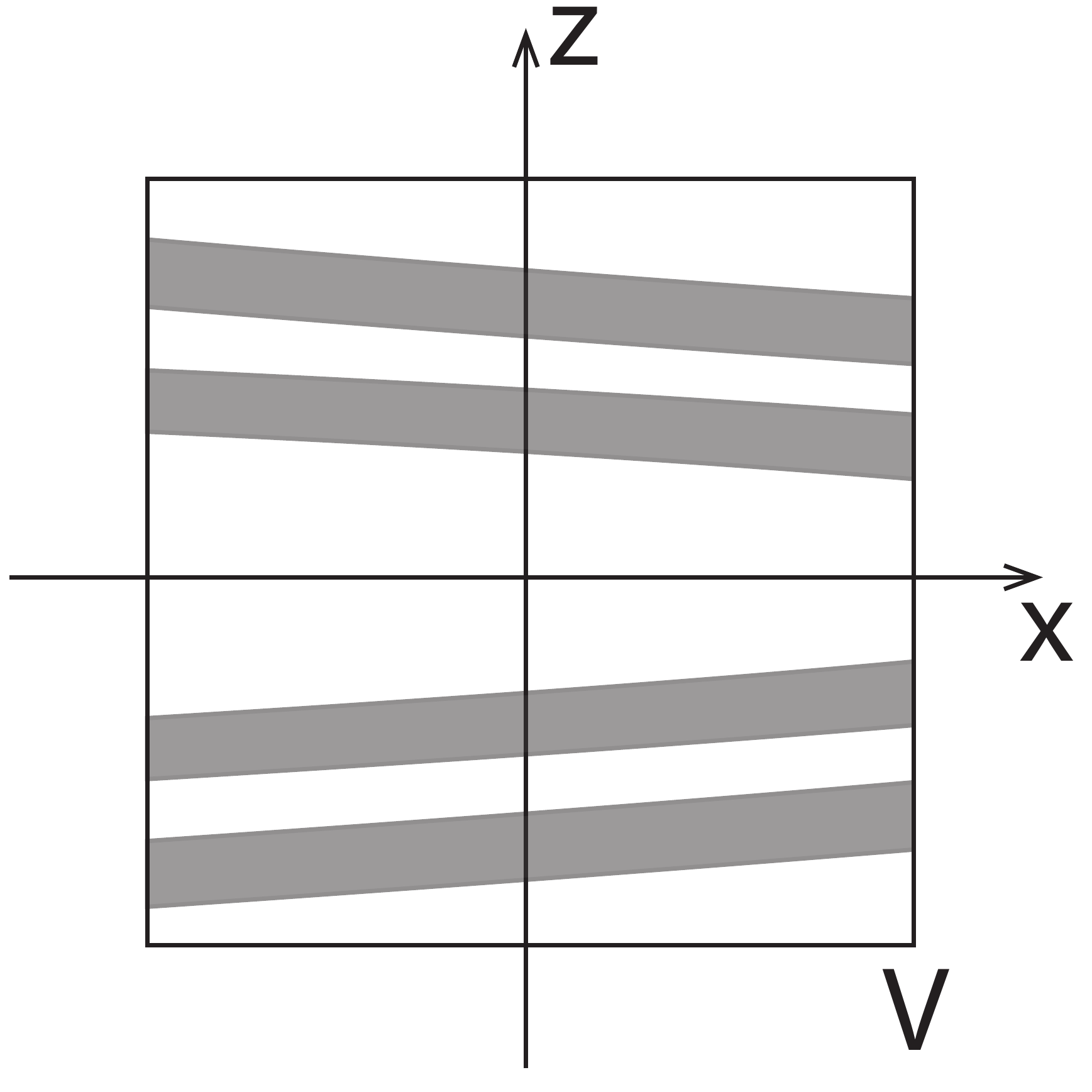}
            \caption[]%
            { $V \cap f_{\rm IV}(V)|_{\Sigma^2(y,w)}$}    
            \label{fig:V_intersect_f_IV_V_x_z_slice}
        \end{subfigure}
        \hfill
        \begin{subfigure}[b]{0.35\textwidth}  
            \centering 
            \includegraphics[width=\textwidth]{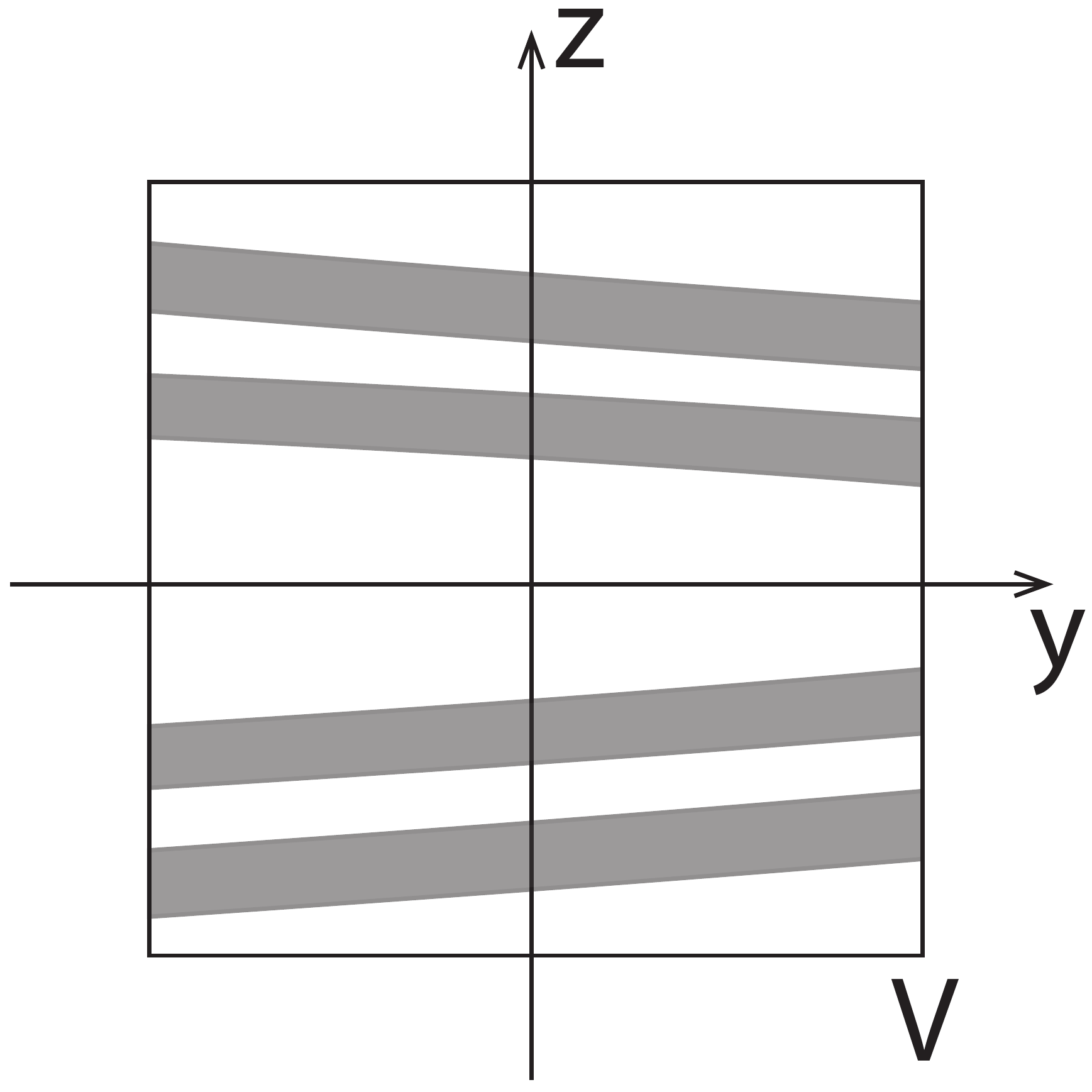}
            \caption[]%
            {$V \cap f_{\rm IV}(V)|_{\Sigma^2(x,w)}$}    
            \label{fig:V_intersect_f_IV_V_y_z_slice}
        \end{subfigure}
        \caption[]
        {(Schematic) Within every $\Sigma^3(w)$ for which $|w| \leq r$, $V \cap f_{\rm IV}(V)$ is topologically equivalent to: (a) four disjoint horizontal strips in every $\Sigma^2(y,w)$ for which $|y|,|w| \leq r$; (b): four disjoint horizontal strips in every $\Sigma^2(x,w)$ for which $|x|,|w| \leq r$. } 
        \label{fig:V_intersect_f_IV_V_xyz_slices}
    \end{figure}
We have already established in Sec.~\ref{Type II} that condition (a) is equivalent to the following two conditions:
\begin{itemize}
\item[(a.1)] $V \cap f_{\rm IV}(V)|_{\Sigma^2(y,w)}$ consists of four disjoint horizontal strips for all $|y|,|w| \leq r$, as shown by Fig.~\ref{fig:V_intersect_f_IV_V_x_z_slice}.
\item[(a.2)] $V \cap f_{\rm IV}(V)|_{\Sigma^2(x,w)}$ consists of four disjoint horizontal strips for all $|x|,|w| \leq r$, as shown by Fig.~\ref{fig:V_intersect_f_IV_V_y_z_slice}.
\end{itemize}
Next, we prove conditions (a.1), (a.2), and (b) individually.  

(a.1): Combining Eq.~(\ref{eq:f_IV inverse z component}) and the third row of Eq.~(\ref{eq:V intersect f IV V}) we obtain
\begin{eqnarray}
& -s = -(z^{\prime})^4 + \left[ 2(a-y^{\prime})-\frac{bc}{b+c} \right](z^{\prime})^2 \nonumber \\
&~~~ + \left[ -(y^{\prime})^2 + \left( 2a-\frac{bc}{b+c} \right)y^{\prime} +a -a^2 + \frac{abc}{b+c} - x^{\prime} - \frac{c}{b+c}w^{\prime} \right]  
\end{eqnarray}
where the parameter $s$ varies within range $-r \leq s \leq r$. Omitting the primes on the variables and expressing $x$ in terms of $(z,y,w,s)$ we get
\begin{eqnarray}\label{eq:Gamma_x in terms of z y w s}
& x = \Gamma_x(z,y,w,s) = -z^4 + \left[ 2(a-y)-\frac{bc}{b+c} \right]z^2 \nonumber \\
& ~~~ -\left[ (y-a)^2+\frac{bc}{b+c} y \right] + \left(1+\frac{bc}{b+c}\right)a +s - \frac{c}{b+c} w \ .
\end{eqnarray}
$\Gamma_x(z,y,w,s)$ is viewed as a quartic function of $z$, parameterized by $(y,w,s)$ with range $|y|,|w|,|s| \leq r$. 

In a particular $\Sigma^2(y,w)$ slice, the values of $(y,w)$ are fixed, the maximum and minimum of $\Gamma_x(z,y,w,s)$ are attained at
\begin{eqnarray}\label{eq:Gamma_x in terms of z y w s max}
& \Gamma^{\max}_x(z,y,w)=\Gamma_x(z,y,w,s)|_{s=r} = -z^4 + \left[ 2(a-y)-\frac{bc}{b+c} \right]z^2 \nonumber \\
& ~~~ -\left[ (y-a)^2+\frac{bc}{b+c} y \right] + \left(1+\frac{bc}{b+c}\right)a +r - \frac{c}{b+c} w \ ,
\end{eqnarray}
\begin{eqnarray}\label{eq:Gamma_x in terms of z y w s min}
& \Gamma^{\min}_x(z,y,w)=\Gamma_x(z,y,w,s)|_{s=-r} = -z^4 + \left[ 2(a-y)-\frac{bc}{b+c} \right]z^2 \nonumber \\
& ~~~ -\left[ (y-a)^2+\frac{bc}{b+c} y \right] + \left(1+\frac{bc}{b+c}\right)a -r - \frac{c}{b+c} w \ .
\end{eqnarray}
Furthermore, the extremals of $\Gamma_x(z,y,w,s)$ within a particular $\Sigma^2(y,w)$ slice can be identified by the condition
\begin{equation}\label{eq:Gamma_x extremal condition}
\frac{\mathrm{d}}{\mathrm{d}z}\Gamma_x(z,y,w,s) = 0
\end{equation}
which leads to
\begin{equation}\label{eq:Gamma_x extremal condition explicit}
-4z^3 + 2\left[ 2(a-y)-\frac{bc}{b+c} \right]z = 0  \ .
\end{equation}
Eq.~(\ref{eq:Gamma_x extremal condition explicit}) has three solutions, which lead to three extremals, namely point $A$, $B$, and $C$, where
\begin{eqnarray}\label{eq:Extremal A}
& A = (x_A,z_A) \nonumber \\
& x_A = -\left[ (y-a)^2 + \frac{bc}{b+c}y \right] + \left( 1+\frac{bc}{b+c} \right)a +s - \frac{c}{b+c} w  \\
& z_A = 0 \ , \nonumber
\end{eqnarray}
\begin{eqnarray}\label{eq:Extremal B}
& B = (x_B,z_B) \nonumber \\
&  x_B = a+s+\frac{b^2c^2-4c(b+c)w}{4(b+c)^2}   \\
& z_B = \sqrt{(a-y)-\frac{bc}{2(b+c)}} \ , \nonumber
\end{eqnarray}
and
\begin{eqnarray}\label{eq:Extremal C}
& C = (x_C,z_C) \nonumber \\
&  x_C = a+s+\frac{b^2c^2-4c(b+c)w}{4(b+c)^2}    \\
& z_C = -\sqrt{(a-y)-\frac{bc}{2(b+c)}} \ . \nonumber
\end{eqnarray}
The second derivatives of $\Gamma_x$ at the three extremals can be computed as
\begin{eqnarray}
& \frac{ \mathrm{d}^2 \Gamma_x}{\mathrm{d}z^2}\Bigg|_A = 4(a-y) -2\frac{bc}{b+c} \geq 4(a-r) -2\frac{bc}{b+c} \ , \\
& \frac{ \mathrm{d}^2 \Gamma_x}{\mathrm{d}z^2}\Bigg|_B =   \frac{ \mathrm{d}^2 \Gamma_x}{\mathrm{d}z^2}\Bigg|_C =  -8(a-y) + 4 \frac{bc}{b+c} \leq -8(a-r) + 4 \frac{bc}{b+c} \ .
\end{eqnarray}
\begin{figure}
        \centering
        \includegraphics[width=0.8\linewidth]{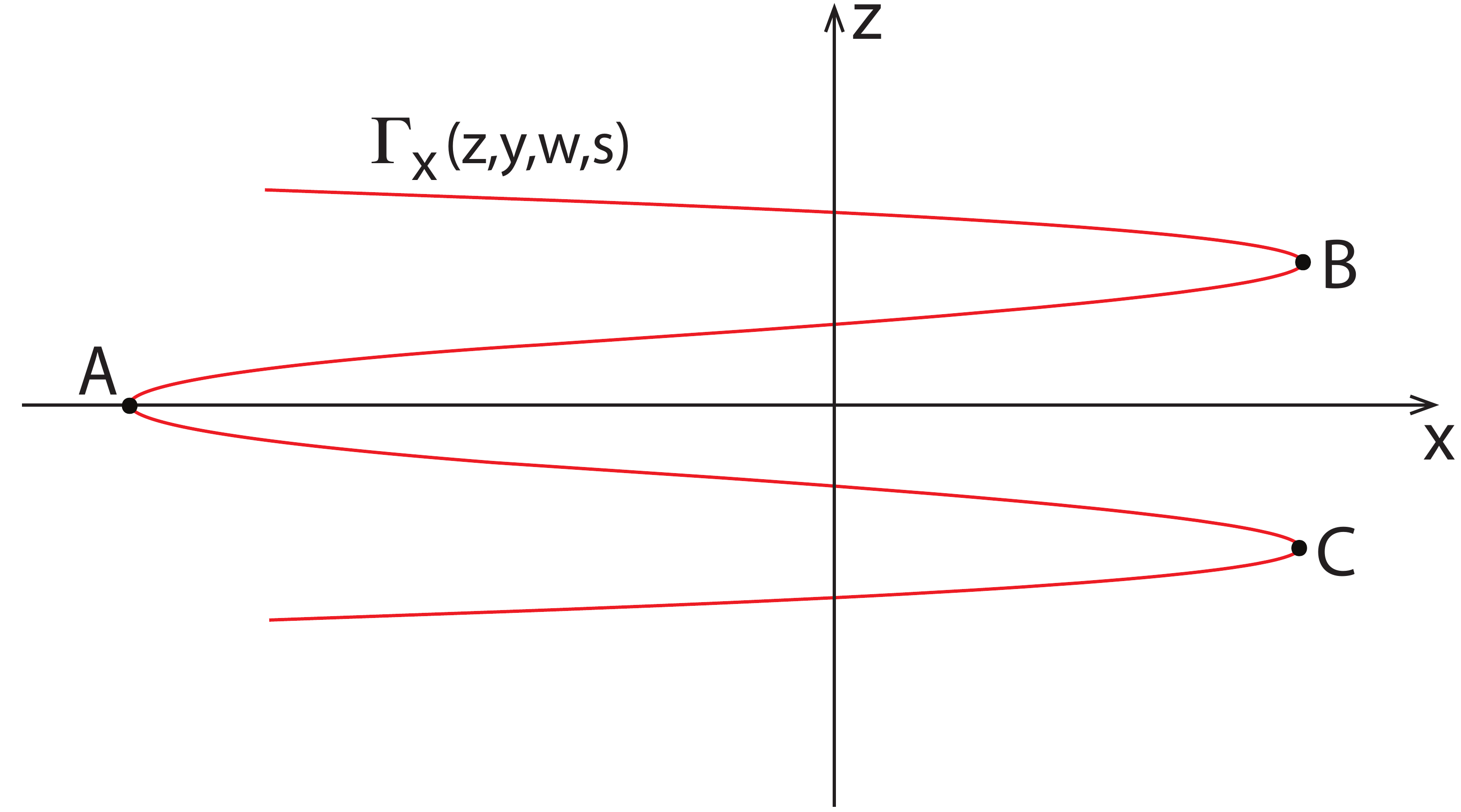} 
        \caption{(Schematic, color online) Within every $\Sigma^2(y,w)$ for which $|y|,|w| \leq r$, $\Gamma_x(z,y,w,s)$ is a quartic function of $z$. It has a local minimum ($A$), and two global maximums ($B$ and $C$). }   \label{fig:Quartic_function_schematic_4D}
\end{figure}
Using the bound imposed by Eq.~(\ref{eq:4D Doubly folded horseshoe with cusp topology parameter bound 2}) we can easily derive that 
\begin{eqnarray}
& \frac{ \mathrm{d}^2 \Gamma_x}{\mathrm{d}z^2}\Bigg|_A >0 \ , \\
& \frac{ \mathrm{d}^2 \Gamma_x}{\mathrm{d}z^2}\Bigg|_B =   \frac{ \mathrm{d}^2 \Gamma_x}{\mathrm{d}z^2}\Bigg|_C < 0 \ .
\end{eqnarray}
This indicates that $A$ is a local minimum, while $B$ and $C$ are two global maximums of the quartic function $\Gamma_x$, as illustrated by Fig.~\ref{fig:Quartic_function_schematic_4D}.

\begin{figure}
        \centering
        \includegraphics[width=0.8\linewidth]{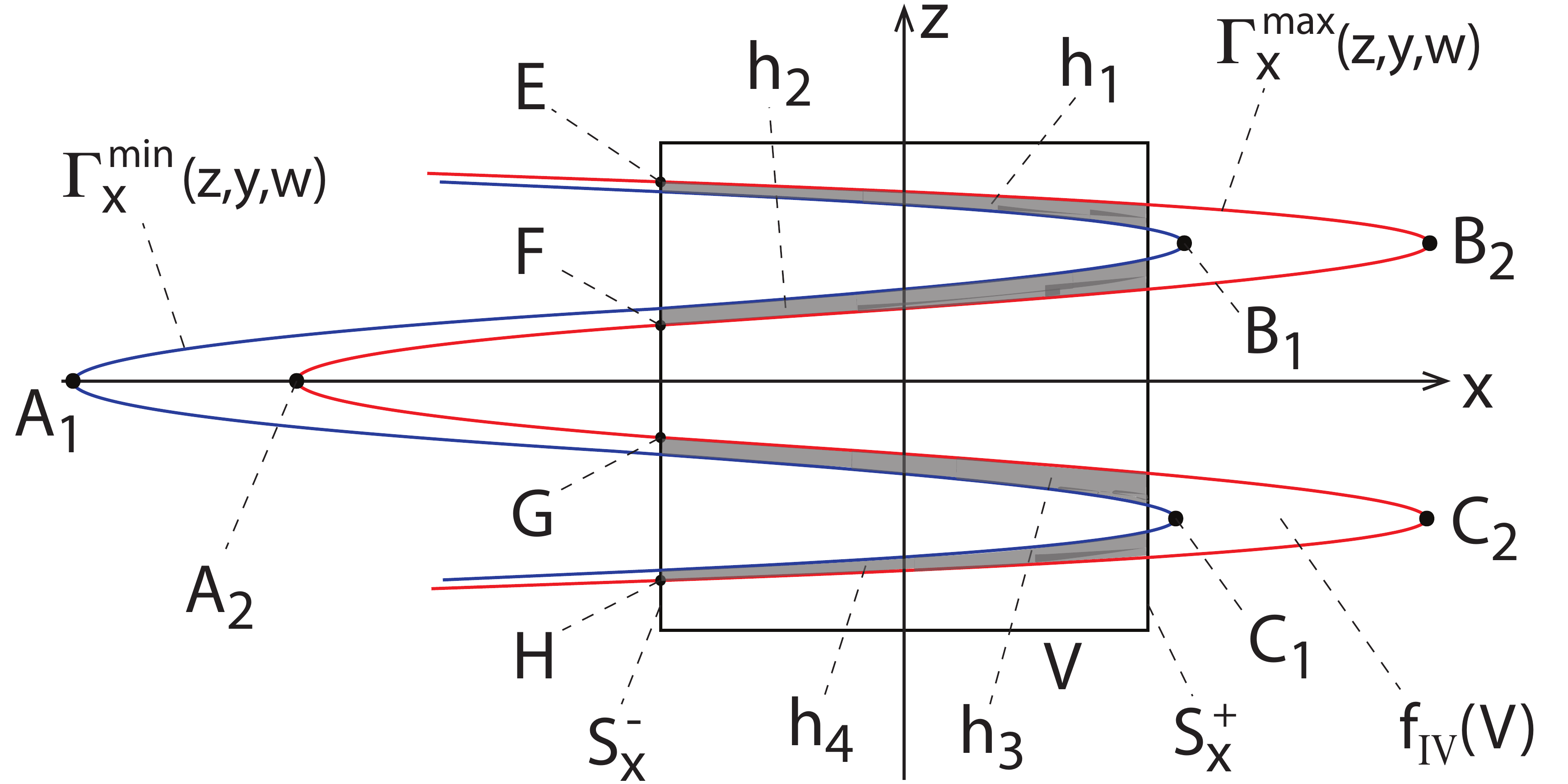} 
        \caption{(Schematic, color online) In every $\Sigma^2(y,w)$ for which $|y|,|w| \leq r$, $f_{\rm IV}(V)$ is given by the gap between $\Gamma^{\min}_x(z,y,w)$ (blue) and $\Gamma^{\max}_x(z,y,w)$ (red). $V \cap f_{\rm IV}(V)$ consists of four disjoint horizontal strips labeled by $h_1$, $h_2$, $h_3$, and $h_4$.}   \label{fig:4D_doubly_folded_horseshoe_cusp_x_z_slice}
\end{figure}

At this point, we have shown that $\Gamma_x$ has the desired shape. Consequently, the following two conditions are sufficient for (a.1):
\begin{enumerate}
\item[(a.1.1)] The global maximums of $\Gamma^{\min}_x(z,y,w)$, as labeled by $B_1$ and $C_1$ in Fig.~\ref{fig:4D_doubly_folded_horseshoe_cusp_x_z_slice}, are on the right-hand side of $S^+_x$;
\item[(a.1.2)] $\Gamma^{\max}_x(z,y,w)$ intersects $S^-_{x}$ at four points, as labeled by E, F, G, H in Fig.~\ref{fig:4D_doubly_folded_horseshoe_cusp_x_z_slice}.
\end{enumerate} 

To establish (a.1.1), note that the global maximums of $\Gamma^{\min}_x(z,y,w)$ are $B_1=(x_{B_1},z_{B_1})$ and $C_1=(x_{C_1},z_{C_1})$, where 
\begin{equation}
x_{B_1} = x_{C_1} = a - r + \frac{b^2 c^2 - 4c(b+c)w}{4(b+c)^2}\ .
\end{equation}
The bound imposed by Eq.~(\ref{eq:4D Doubly folded horseshoe with cusp topology parameter bound 3}) guarantees that
\begin{equation}
b^2 \left( \frac{c}{b+c} \right)^2 -4r \left( \frac{c}{b+c} \right) + 4(a-2r) > 0
\end{equation}
which can be simplified into
\begin{equation}
\frac{cr}{b+c} - \frac{b^2c^2}{4(b+c)^2} < a-2r\ .
\end{equation}
Since $|w|\leq r$ we have
\begin{equation}\label{eq:B1 C1 horizontal position criterion}
\frac{cw}{b+c} - \frac{b^2c^2}{4(b+c)^2} < \frac{cr}{b+c} - \frac{b^2c^2}{4(b+c)^2} <a-2r\ ,
\end{equation}
which, upon some algebraic manipulations, yields
\begin{equation}
a - r + \frac{b^2 c^2 - 4c(b+c)w}{4(b+c)^2} > r \ ,
\end{equation}
i.e., $x_{B_1}=x_{C_1} > r$. Thus (a.1.1) is established. 

To establish (a.1.2), it is sufficient to establish the following two conditions:
\begin{enumerate}
\item[(a.1.2.1)] The local minimum of $\Gamma^{\max}_x(z,y,w)$, labeled by $A_2$ in Fig.~\ref{fig:4D_doubly_folded_horseshoe_cusp_x_z_slice}, is located on the left-hand side of $S^-_x$\ .
\item[(a.1.2.2)] $\Gamma^{\max}_x(z,y,w)\Big|_{z=\pm r} \leq -r$ for all $|y|,|w| \leq r$. 
\end{enumerate}

\begin{figure}
        \centering
        \includegraphics[width=0.5\linewidth]{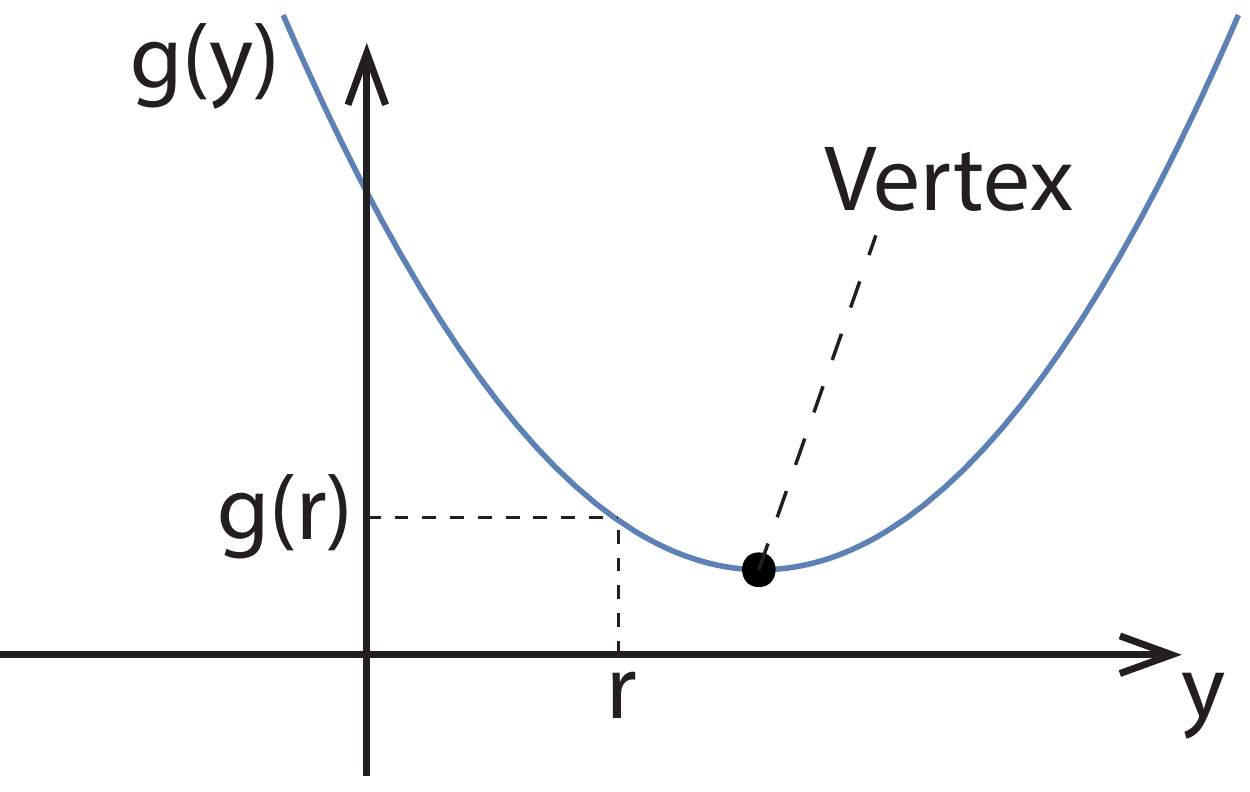} 
        \caption{(Schematic) The vertex of the parabola $g(y)$ is determined by Eq.~(\ref{eq:Vertex of g(y)}) . The minimum of $g(y)$ within range $|y| \leq r$ is attained at $y=r$. }   \label{fig:Parabola_Schematic}
\end{figure}

To prove (a.1.2.1), note that the horizontal position of $A_2$ is 
\begin{equation}x_{A_2} = -\left[ (y-a)^2 + \frac{bc}{b+c}y \right] + \left( 1+\frac{bc}{b+c} \right)a +r - \frac{c}{b+c} w \ .
\end{equation}
Let 
\begin{equation}
g(y) \equiv (y-a)^2 + \frac{bc}{b+c}y = y^2 - \left( 2a-\frac{bc}{b+c} \right)y + a^2\ , 
\end{equation}
then $x_{A_2}$ is re-written as
\begin{equation}
x_{A_2} = -g(y) + \left( 1+\frac{bc}{b+c} \right)a +r - \frac{c}{b+c} w \ .
\end{equation}
The vertex of the parabola $g(y)$ is 
\begin{equation}\label{eq:Vertex of g(y)}
\mathrm{Vertex~of~} g(y) = \left( a-\frac{bc}{2(b+c)}, \frac{4a^2 - (2a-\frac{bc}{b+c})^2}{4} \right)\ ,
\end{equation}
as demonstrated by Fig.~\ref{fig:Parabola_Schematic}. From the bound imposed by Eq.~(\ref{eq:4D Doubly folded horseshoe with cusp topology parameter bound 2}) we know
\begin{equation}
a-\frac{bc}{2(b+c)}>r\ ,
\end{equation}
thus the vertex of $g(y)$ is not attained within $|y| \leq r$. Therefore, the minimum of $g(y)$ within $|y|\leq r$ is attained at $y=r$:
\begin{equation}
g(y) \geq g(r) = (a-r)^2 + \frac{bc}{b+c}r,~~~\forall~|y|\leq r\ ,
\end{equation}
as illustrated by Fig.~\ref{fig:Parabola_Schematic}. Consequently, $x_{A_2}$ is bounded from above by 
\begin{eqnarray}\label{eq:x A2 upper bound}
& x_{A_2} \leq -g(r) + \left( 1+\frac{bc}{b+c} \right)a + r + \frac{c}{b+c}r \nonumber \\
& ~~~~~ = -(a-r)^2 - \frac{bc}{b+c}r + \left( 1+\frac{bc}{b+c} \right)a + r + \frac{c}{b+c}r \ .
\end{eqnarray}
From the bound imposed by Eq.~(\ref{eq:4D Doubly folded horseshoe with cusp topology parameter bound 4}) we obtain
\begin{equation}\label{eq:c over b+c}
\frac{c}{b+c} < \frac{a(a-2r)}{ab+(1-b)r} < \frac{(a-r)^2-a-2r}{ab+(1-b)r}
\end{equation}
where the second inequality comes from the fact that $r^2 - 2r -a >0$. Eq.~(\ref{eq:c over b+c}) is equivalent to 
\begin{equation}
-(a-r)^2 - \frac{bc}{b+c}r + \left( 1+\frac{bc}{b+c} \right)a + r + \frac{c}{b+c}r < -r \ ,
\end{equation}
which, when substituted into Eq.~(\ref{eq:x A2 upper bound}), yields 
\begin{equation}
x_{A_2} < -r \ .
\end{equation}
Condition (a.1.2.1) is thus established. 

To prove (a.1.2.2), notice that upon straightforward algebraic manipulations, $\Gamma^{\max}_x(z,y,w)\Big|_{z=\pm r} \leq -r$ can be shown to be equivalent to
\begin{eqnarray}\label{eq:long horrible shit}
& h(y) \equiv y^2 + \left( 14a + \frac{bc}{b+c} + 8\sqrt{2}r \right)y + 49a^2 + \frac{8abc}{b+c} \nonumber \\
& ~~~ - a\left( 1+\frac{bc}{b+c} \right) -2r + 56\sqrt{2}ar + \frac{4\sqrt{2}rbc}{b+c} + 32r^2 + \frac{c}{b+c}w \geq 0
\end{eqnarray}
where the facts that $r^2 = 4\sqrt{2}r + 8a$ and $r^4 = 32 r^2 + 64\sqrt{2}ar + 64 a^2$ have been used. Treating $h(y)$ as a quadratic function of $y$ whose graph gives rise to a parabola, the vertex of the parabola is attained at 
\begin{equation}
y_{\rm vertex} = -7a - 4\sqrt{2}r - \frac{bc}{2(b+c)} < -r\ .
\end{equation}
Consequently, the minimum of $h(y)$ within $y \in [-r,r]$ is attained at $y=-r$, thus
\begin{equation}\label{eq:h(y) lower bound}
h(y)\Big|_{y\in [-r,r]} \geq h(-r)\ .
\end{equation}
Furthermore, $h(-r)$ can be shown to be positive. This is due to the fact that
\begin{eqnarray}\label{eq:longer horrible shit}
& h(-r) = r^2 - \left( 14a + \frac{bc}{b+c} + 8\sqrt{2}r \right)r + 49a^2 + \frac{8abc}{b+c}\nonumber \\
& ~~~ - a\left( 1+\frac{bc}{b+c} \right) -2r + 56\sqrt{2}ar + \frac{4\sqrt{2}rbc}{b+c} + 32r^2 + \frac{c}{b+c}w \nonumber \\
& \geq r^2 - \left( 14a + \frac{bc}{b+c} + 8\sqrt{2}r \right)r + 49a^2 + \frac{8abc}{b+c}\nonumber \\
& ~~~ - a\left( 1+\frac{bc}{b+c} \right) -2r + 56\sqrt{2}ar + \frac{4\sqrt{2}rbc}{b+c} + 32r^2 - \frac{c}{b+c}r \nonumber \\
& = \left[ (4\sqrt{2}-1)br + 7ab - r \right]\frac{c}{b+c} + (132\sqrt{2} - 66)r  \nonumber \\
& ~~~+ (263-64\sqrt{2})a + (56\sqrt{2}-14)ar + 49a^2 >0
\end{eqnarray}
where the fact that $r^2 = 4\sqrt{2}r + 8a$ is again used. Combining Eqs.~(\ref{eq:h(y) lower bound}) and (\ref{eq:longer horrible shit}) we get
\begin{equation}
h(y)\Big|_{y\in[-r,r]} > 0\ ,
\end{equation}
which is equivalent to 
\begin{equation}
\Gamma^{\max}_x(z,y,w)\Big|_{z=\pm r} \leq -r\ .
\end{equation}
Thus (a.1.2.2) is established. Consequently, conditions (a.1.2) and (a.1) are established as well.

(a.2): from the second row in Eq.~(\ref{eq:V intersect f IV V}) we know
\begin{equation}
-y-z^2+a=-s
\end{equation}
where $|s| \leq r$ and the primes on the variables are omitted for simplicity. Therefore 
\begin{equation}
y = -z^2+a+s
\end{equation}
which is identical to Eq.~(\ref{eq:y z parabola 1 f_II}). Therefore, repeating the same steps in Section.~\ref{Type II}, we define the parabolas $\Gamma_y(z,s)$, $\Gamma^{\min}_y(z)$, and $\Gamma^{\max}_y(z)$ using Eq.~(\ref{eq:Gamma y in terms of z s f_II}), (\ref{eq:Gamma y z min}) and (\ref{eq:Gamma y z max}), respectively, and the relationship in Eq.~(\ref{eq:Gamma y in terms of z s bounds}) holds immediately. 

Within each $\Sigma^2(x,w)$ with $|x|,|w| \leq r$, $f_{\rm IV}(V)$ lies within the gap  between the two parabolas $\Gamma^{\min}_y(z)$ and $\Gamma^{\max}_y(z)$. The locations of the two parabolas can be further narrowed down by establishing the following facts:
\begin{itemize}
\item[(a.2.1)] The vertex of $\Gamma^{\min}_y(z)$ is located on the right-hand side of $S^+_y$, labeled by $A$ in Fig.~\ref{fig:3D_dooubly_folded_horseshoe_y_z_slice_Gammas};
\item[(a.2.2)] $\Gamma^{\max}_y(z)$ intersects $S^-_y$ at two points, labeled by $C$ and $D$ in Fig.~\ref{fig:3D_dooubly_folded_horseshoe_y_z_slice_Gammas}. 
\end{itemize} 

To establish (a.2.1), let $A=(y_A,z_A)$. It can be solved easily that
\begin{equation}
z_A=0, ~~~~ y_A = \Gamma^{\min}_y(z_A=0)=a-r. 
\end{equation}
Recall that $a-2r>0$, thus $y_A >r$, i.e., $A$ is on the right-hand side of $S^{+}_y$. 

To establish (a.2.2), notice that since $r = 2\sqrt{2}(1+1\sqrt{1+a})>1+1\sqrt{1+a}$,
\begin{equation}
\Gamma^{\max}_y(z=\pm r) = -r^2 +a +r < -r\ ,
\end{equation}
thus (a.2.2) follows. Combining conditions (a.2.1) and (a.2.2), we obtain
\begin{equation}\label{eq:V intersect f_IV V y z slice first bound}
V \cap f_{\rm IV}(V)|_{\Sigma^2(x,w)} \subset H_1 \cup H_2
\end{equation}
where $H_1$ and $H_2$ are horizontal strips shown by Fig.~\ref{fig:3D_dooubly_folded_horseshoe_y_z_slice_Gammas}. 

At this point, Eq.~(\ref{eq:V intersect f_IV V y z slice first bound}) provides a coarse bound for $V \cap f_{\rm IV}(V)|_{\Sigma^2(x,w)}$. A refined bound can be obtained in the following. Using the third row of Eq.~(\ref{eq:V intersect f IV V}) along with Eq.~(\ref{eq:f_IV inverse z component}) we get:
\begin{eqnarray}\label{eq:f_IV inverse z component parameter s}
& -s =-(z^{\prime})^4 + \left[ 2(a-y^{\prime})-\frac{bc}{b+c} \right](z^{\prime})^2 \nonumber \\
&~~~ + \left[ -(y^{\prime})^2 + \left( 2a-\frac{bc}{b+c} \right)y^{\prime} +a -a^2 + \frac{abc}{b+c} - x^{\prime} - \frac{c}{b+c}w^{\prime} \right] 
\end{eqnarray}
where $s$ is a parameter within range $-r \leq s \leq r$. Omitting the primes on the variables and expressing $y$ in terms of $(z,x,w,s)$ lead to two branches of solutions 
\begin{eqnarray}\label{eq:f IV y in terms of z x w s}
&\Lambda^{\pm}_y(z,x,w,s) = y_{\pm}(z,x,w,s) \nonumber \\
& = -z^2 + a - \frac{bc}{2(b+c)} \pm  \sqrt{(a-x+s)-\frac{cw}{b+c} + \frac{b^2 c^2}{4(b+c)^2} }\ .
\end{eqnarray}
Notice that for the $c=0$ special cases, $f_{\rm IV}$ reduces to the uncoupled direct product between $f_{\rm II}$ in $(x,y,z)$ and a uniform contraction in $(w)$, correspondingly, Eq.~(\ref{eq:f IV y in terms of z x w s}) reduces to Eq.~(\ref{eq:Lambda y in terms of z x s}). 

Upon varying the values of $(x,w,s)$ within range $|x|,|w|,|s| \leq r$ while keeping $z$ fixed, the maximum and minimum of Eq.~(\ref{eq:f IV y in terms of z x w s}) are attained at
\begin{eqnarray}
&\Lambda^{+,\max}_y(z) = \Lambda^{+}_y(z,x,w,s)\Big|_{x=-r,~w=-r,~s=r} \nonumber \\
& = -z^2 + a - \frac{bc}{2(b+c)} +  \sqrt{(a+2r)+\frac{cr}{b+c} + \frac{b^2 c^2}{4(b+c)^2} } \label{eq:Gamma y + max f IV} \\
&\Lambda^{+,\min}_y(z) = \Lambda^{+}_y(z,x,w,s)\Big|_{x=r,~w=r,~s=-r} \nonumber \\
& = -z^2 + a - \frac{bc}{2(b+c)} +  \sqrt{(a-2r)-\frac{cr}{b+c} + \frac{b^2 c^2}{4(b+c)^2} }  \label{eq:Gamma y + min f IV} \\
&\Lambda^{-,\max}_y(z) = \Lambda^{-}_y(z,x,w,s)\Big|_{x=r,~w=r,~s=-r} \nonumber \\
& = -z^2 + a - \frac{bc}{2(b+c)} -  \sqrt{(a-2r)-\frac{cr}{b+c} + \frac{b^2 c^2}{4(b+c)^2} }  \label{eq:Gamma y - max f IV} \\
&\Lambda^{-,\min}_y(z) = \Lambda^{-}_y(z,x,w,s)\Big|_{x=-r,~w=-r,~s=r} \nonumber \\
& = -z^2 + a - \frac{bc}{2(b+c)} -  \sqrt{(a+2r)+\frac{cr}{b+c} + \frac{b^2 c^2}{4(b+c)^2} }  \label{eq:Gamma y - min f IV} \ .
\end{eqnarray}
It can be verified that Eq.~(\ref{eq:4D Doubly folded horseshoe with cusp topology parameter bound 3}) guarantees
\begin{equation}
(a-2r)-\frac{cr}{b+c} + \frac{b^2 c^2}{4(b+c)^2} > 0\ ,
\end{equation}
thus the square roots in Eqs.~(\ref{eq:Gamma y + min f IV}) and (\ref{eq:Gamma y - max f IV}) are real-valued. It is quite obvious that
\begin{equation}\label{eq:f IV y z slice four parabolas ranking}
\Lambda^{-,\min}_y(z)<\Lambda^{-,\max}_y(z)<\Lambda^{+,\min}_y(z)<\Lambda^{+,\max}_y(z)\ .
\end{equation}

\begin{figure}
        \centering
        \includegraphics[width=0.9\linewidth]{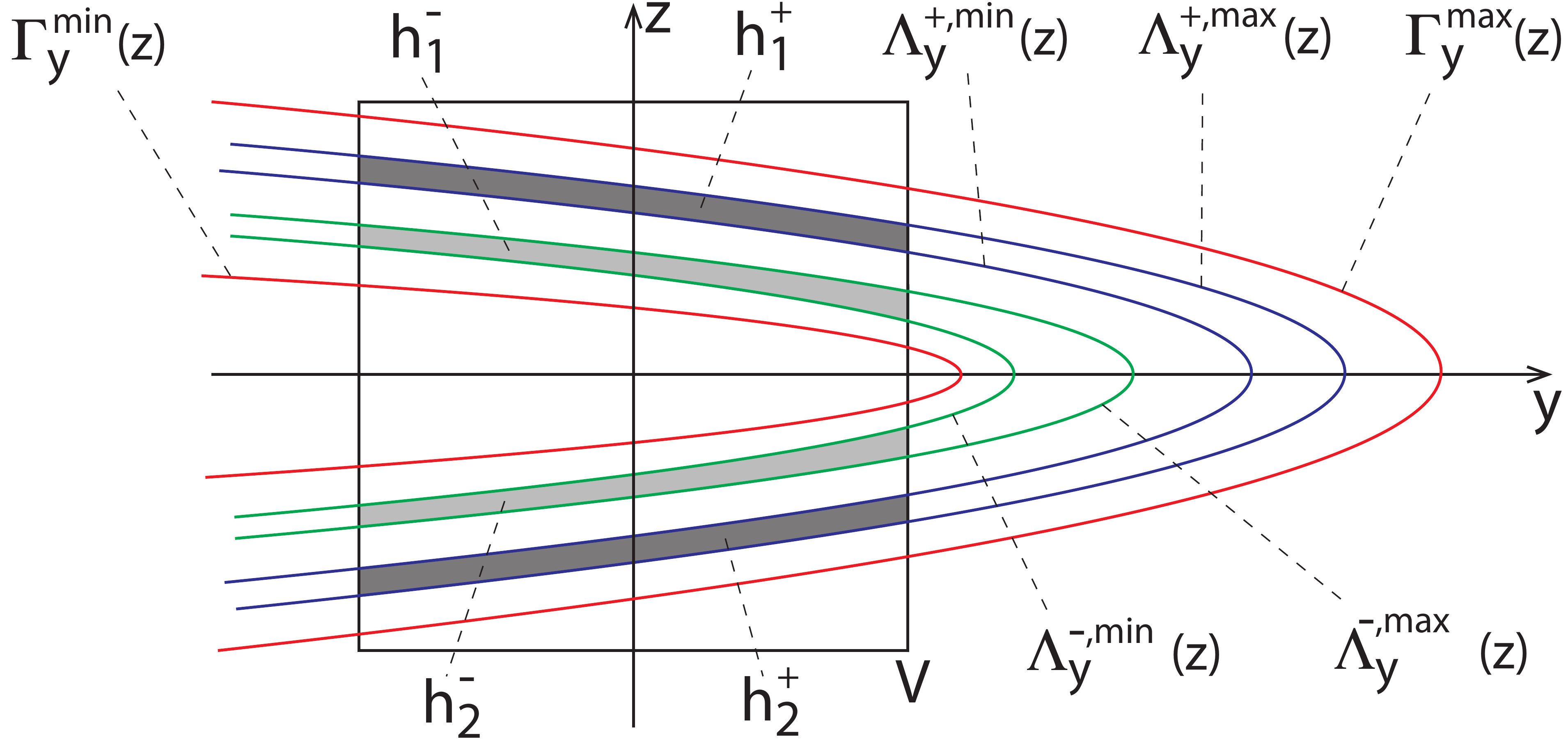} 
        \caption{(Schematic, color online) In every $\Sigma^2(x,w)$ for which $|x|,|w| \leq r$, $V \cap f_{\rm IV}(V)$ consists of four disjoint horizontal strips, $h^{\pm}_1$ and $h^{\pm}_2$, as indicated by the shaded regions. } \label{fig:4D_doubly_folded_horseshoe_cusp_y_z_slice_six_parabolas}
\end{figure}

Recall that we have previously obtained a coarse bound for $V \cap f_{\rm IV}(V)|_{\Sigma^2(x,w)}$ in Eq.~(\ref{eq:V intersect f_IV V y z slice first bound}), where the horizontal strips $H_1$ and $H_2$ are bounded by parabolas $\Gamma^{\min}_y(z)$ and $\Gamma^{\max}_y(z)$ defined by Eq.~(\ref{eq:Gamma y z min}) and (\ref{eq:Gamma y z max}), respectively. We now compare the horizontal positions of the six parabolas ($\Gamma^{\max}_y(z)$, $\Gamma^{\min}_y(z)$, $\Lambda^{\pm,\max}_y(z)$, and $\Lambda^{\pm,\min}_y(z)$), and we would like to show that under the assumptions of the theorem,
\begin{eqnarray}
& \Gamma^{\min}_y(z) < \Lambda^{-,\min}_y(z) \label{eq:f IV y z slice 6 parabolas left boundary condition} \\
& \Lambda^{+,\max}_y(z) < \Gamma^{\max}_y(z) \label{eq:f IV y z slice 6 parabolas right boundary condition} \ .
\end{eqnarray}

To prove Eq.~(\ref{eq:f IV y z slice 6 parabolas right boundary condition}), notice that it is equivalent to
\begin{equation}
- \frac{bc}{2(b+c)} +  \sqrt{(a+2r)+\frac{cr}{b+c} + \frac{b^2 c^2}{4(b+c)^2} } < r\ ,
\end{equation}
which can be simplified into
\begin{equation}
\frac{c}{b+c} < \frac{r^2-2r-a}{(1-b)r}\ . \label{eq:f IV y z slice 6 parabolas right boundary condition simplified}
\end{equation}
Since Eq.~(\ref{eq:f IV y z slice 6 parabolas right boundary condition simplified}) is guaranteed by Eq.~(\ref{eq:4D Doubly folded horseshoe with cusp topology parameter bound 5}), Eq.~(\ref{eq:f IV y z slice 6 parabolas right boundary condition}) holds. 

To prove Eq.~(\ref{eq:f IV y z slice 6 parabolas left boundary condition}), notice that it is equivalent to 
\begin{equation}
-r < - \frac{bc}{2(b+c)} -  \sqrt{(a+2r)+\frac{cr}{b+c} + \frac{b^2 c^2}{4(b+c)^2} } \ ,
\end{equation}
which can be further simplified into 
\begin{equation}
\frac{c}{b+c} < \frac{r^2-2r-a}{(1+b)r}\ , \label{eq:f IV y z slice 6 parabolas left boundary condition simplified}
\end{equation}
which is identical to Eq.~(\ref{eq:4D Doubly folded horseshoe with cusp topology parameter bound 5}). Thus Eq.~(\ref{eq:f IV y z slice 6 parabolas left boundary condition}) is established as well. Combining Eqs.~(\ref{eq:f IV y z slice four parabolas ranking}), (\ref{eq:f IV y z slice 6 parabolas left boundary condition}) and (\ref{eq:f IV y z slice 6 parabolas right boundary condition}) we obtain
\begin{equation}\label{eq:f IV y z slice six parabolas ranking}
\Gamma^{\min}_y(z) <\Lambda^{-,\min}_y(z)<\Lambda^{-,\max}_y(z)<\Lambda^{+,\min}_y(z)<\Lambda^{+,\max}_y(z) < \Gamma^{\max}_y(z)
\end{equation}
i.e., the six parabolas are positioned from left to right as illustrated by Fig.~\ref{fig:4D_doubly_folded_horseshoe_cusp_y_z_slice_six_parabolas}. Therefore we obtain the final expression 
\begin{equation}\label{eq:V intersect f_IV V y z slice final expression}
V \cap f_{\rm IV}(V)|_{\Sigma^2(x,w)} = h^+_1 \cup h^+_2 \cup h^-_1 \cup h^-_2 \subset H_1 \cup H_2
\end{equation}
where the four horizontal strips $h^{\pm}_1$ and $h^{\pm}_2$ are indicated by the shaded regions in Fig.~\ref{fig:4D_doubly_folded_horseshoe_cusp_y_z_slice_six_parabolas}. Condition (a.2) is thus established. Consequently, condition (a) is established as well. 

(b): Let $V \cap f_{\rm IV}(V)\Big|_{\Sigma^1(x,y,z)}$ be the restriction of $V \cap f_{\rm IV}(V) \cap \Sigma^1(x,y,z)$ on $\Sigma^1(x,y,z)$. We need to show that $V \cap f_{\rm IV}(V)\Big|_{\Sigma^1(x,y,z)}$ is a line segment which is a proper subset of $V|_{\Sigma^1(x,y,z)}$, i.e., the configuration illustrated by Fig.~\ref{fig:f_IV_uncoupled_w_slice}. 

From Eq.~(\ref{eq:V intersect f IV V}) we know that $|w(y^{\prime},z^{\prime},w^{\prime})|\leq r$. This is equivalent to setting $w(y^{\prime},z^{\prime},w^{\prime})=-s$ with parameter $|s| \leq r$. Substituting Eq.~(\ref{eq:f_IV inverse w component}) we get
\begin{equation}
w'(s)=cy' - ca + c(z')^2 - s(b+c)
\end{equation}
or simply 
\begin{equation}\label{eq:w in terms of y z s}
w(s)=cy - ca + cz^2 - s(b+c)
\end{equation}
in which the primes on the variables are omitted for simplicity. In every $\Sigma^1(x,y,z)$, $w(s)$ is interpreted as a one-dimensional family of points parameterized by $s$ with fixed $(x,y,z)$ values. Let $l$ be the line segment such that 
\begin{equation}\label{eq:l_1}
l = \left\lbrace w(s) \middle| -r \leq s \leq r \right\rbrace\ .
\end{equation}
Then combining Eqs.~(\ref{eq:V intersect f IV V}) and (\ref{eq:w in terms of y z s}) we obtain 
\begin{equation}\label{eq:V intersect f_IV V w slice expression}
V \cap f_{\rm IV}(V) \Big|_{\Sigma^1(x,y,z)} = l \ .
\end{equation}
The two endpoints of $l$ are attained at
\begin{eqnarray}
& w^{\max}=w(s=-r)=cy - ca + cz^2 + r(b+c) \nonumber  \\
& w^{\min}=w(s=r)=cy - ca + cz^2 - r(b+c)\ . \nonumber 
\end{eqnarray}
Notice that $w^{\max}$ is bounded from above by 
\begin{equation}\label{eq:f IV w max upper bound}
w^{\max} \leq c(r^2+2r-a)+br \leq r\ ,
\end{equation}
where the second inequality comes from Eq.~(\ref{eq:4D Doubly folded horseshoe with cusp topology parameter bound 1}). Meanwhile, $w^{\min}$ is bounded from below by
\begin{equation}\label{eq:f IV w min lower bound}
w^{\min} \geq -c(a+2r)-br > -r
\end{equation}
where the second equality comes from Eq.~(\ref{eq:4D Doubly folded horseshoe with cusp topology parameter bound 1}) and the fact that $r^2+2r-a>a+2r$. Therefore in every $\Sigma^1(x,y,z)$ we have
\begin{equation}\label{eq:V intersect f_IV V w slice final range}
V \cap f_{\rm IV}(V) \Big|_{\Sigma^1(x,y,z)} \subset (-r,r]
\end{equation}
which proves (b). Having established both (a) and (b), Theorem~\ref{4D Doubly folded horseshoe with cusp topology} is proved. 

\end{proof}

To demonstrate that the combination of parameters $(a,b,c)$ that satisfies the bounds imposed by Theorem~\ref{4D Doubly folded horseshoe with cusp topology} indeed exists, Fig.~\ref{fig:Parameter_Space_Plot} illustrates the domain in the three-dimensional parameter space $(a,b,c)$ that satisfies Eqs.~(\ref{eq:4D Doubly folded horseshoe with cusp topology parameter bound 1})-(\ref{eq:4D Doubly folded horseshoe with cusp topology parameter bound 5}) and $b>c$, within range $(a,b,c) \in [32+8\sqrt{2},42+8\sqrt{2}] \times [0,1] \times [0,0.04]$. When $(a,b,c)$ lies within the colored domain, $f_{\rm IV}(V)$ exhibits the topological structure shown by Fig.~\ref{fig:f_IV_uncoupled}. 

\begin{figure}
        \centering
        \includegraphics[width=0.6\linewidth]{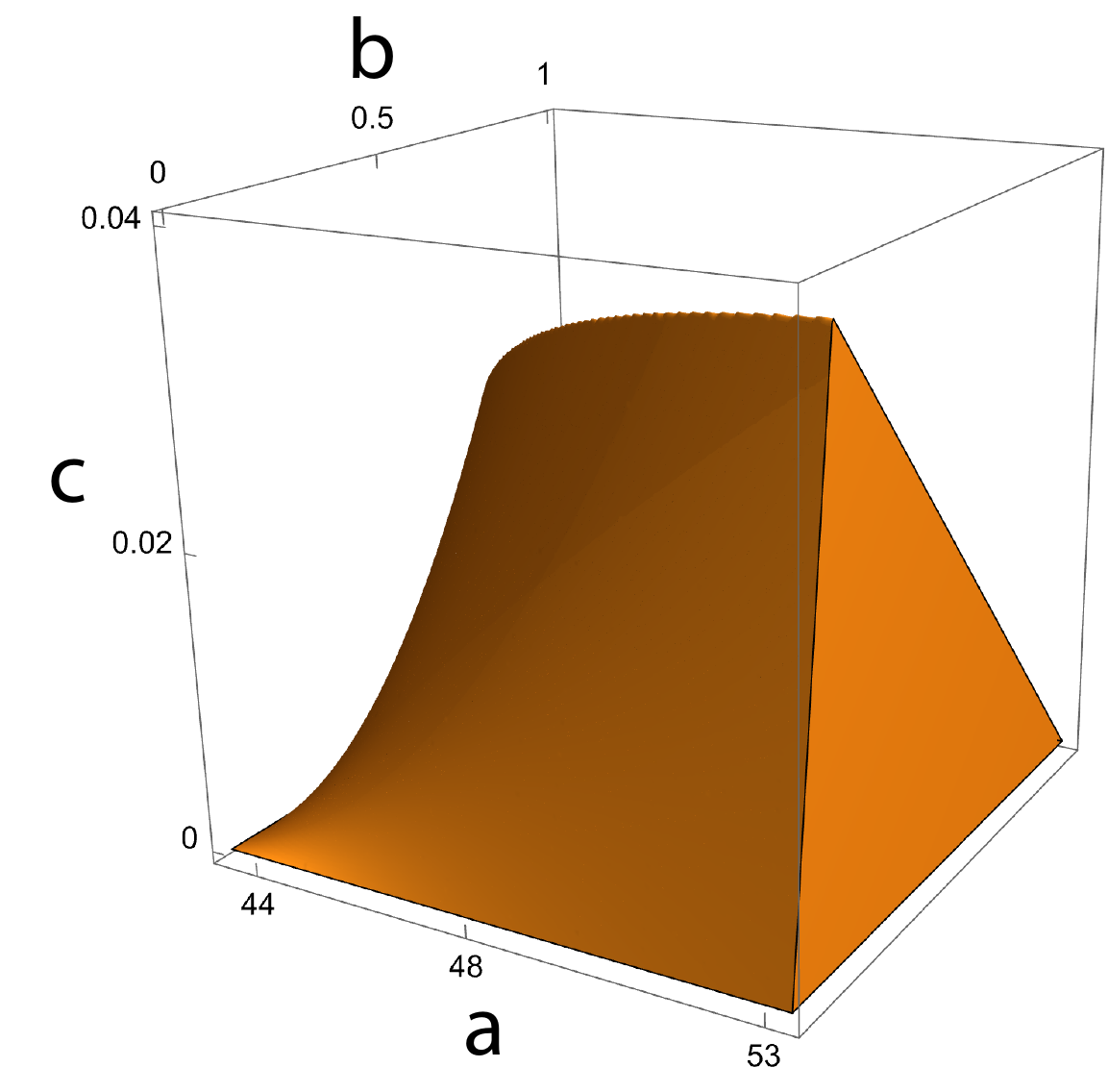} 
        \caption{(Color online) Three-dimensional domain (orange) in the parameter space $(a,b,c)$ that satisfies Eqs.~(\ref{eq:4D Doubly folded horseshoe with cusp topology parameter bound 1})-(\ref{eq:4D Doubly folded horseshoe with cusp topology parameter bound 5}) and $b>c$, within cube $[32+8\sqrt{2},42+8\sqrt{2}] \times [0,1] \times [0,0.04]$. } \label{fig:Parameter_Space_Plot}
\end{figure}

\section{Conclusion}\label{Conclusion}

We have introduced several examples using H\'{e}non-type mappings in three and four dimensions and demonstrated analytically that they possess horseshoe structures with nontrivial folding topologies that are otherwise impossible in two dimensions. In essence, we have designed these maps so that they give rise to different combinations of folding and stacking directions. More specifically, Topology II and IV are twice-folded horseshoes with independent creases but a common stacking direction, Topology III-A is a twice-folded horseshoe with independent creases and independent stacking directions, and Topology III-B is just a once-folded horseshoe in four dimensions. Interestingly, when the parameters of $f_{\rm III}$ are changed from the neighborhood of the Type-A AI limit to the neighborhood of the Type-B AI limit, the twice-folded horseshoe (Topology III-A) undergoes an unfolding process to transition into the once-foled horseshoe (Topology III-B). 

It is obvious that the topological structures introduced here only represent a small fraction of all possible horseshoe topologies. In higher dimensions, more combinations of folding and stacking directions can be selected to create more complicated horseshoe structures, the exploration of which may be potentially fruitful. Furthermore, it is also meaningful to ask the question of 
which type of horseshoe topology is generic for certain types of multidimensional maps, 
and what are the consequent implications on the symbolic dynamics of the systems. 
In particular, within the symplectic setting, according to the paper of Moser \cite{moser1994quadratic}, 
there is a normal form for quadratic symplectic maps with parameters on ${\Bbb R}^2$. 
It is therefore interesting to explore how the horseshoe with different topologies 
coexist in the parameter space of the normal form, and it then automatically raises 
the question what types of bifurcations occur among
different types of horseshoes found here. 
Notice that all these problems do not exist in the horseshoe formed in the 2-dimensional plane. 

Regarding uniform hyperbolicity, we have shown in Part I \cite{Fujioka23} that 
the map forming Topology III-A and III-B horseshoes has parameter regions in which 
not only topological horseshoe but also uniform hyperbolicity holds, 
thus the conjugacy to the symbolic dynamics follows. 
It is then natural to confirm that the same is true as well for the rest of cases. 
The computer-assisted proof \cite{arai2007hyperbolic} 
will be helpful to more precisely specify the region where 
 uniform hyperbolicity holds, and hence to clarify how different types of 
 topological horseshoe coexist in their parameter space. 
Those are topics for our futures studies.

\section*{Acknowledgement}
J.L. gratefully acknowledges many inspiring discussions with Steven Tomsovic. J.L. and A.S. acknowledge financial support from Japan Society for the Promotion of Science (JSPS) through JSPS Postdoctoral Fellowships for Research in Japan (Standard). This work has been supported by JSPS KAKENHI Grant No. 17K05583, and also by JST, the establishment of university fellowships towards the creation of science technology innovation, Grant Number JPMJFS2139.

\section*{References}
\bibliographystyle{iopart-num}
\bibliography{classicalchaos}
\end{document}